\documentclass{article}
\usepackage{amsmath,amssymb,amsfonts,epsfig,verbatim}
\usepackage{times}
\usepackage{color}
\usepackage{ifthen}
\usepackage{calc}
\newif\ifhyper\IfFileExists{hyperref.sty}{\hypertrue}{\hyperfalse}
\hypertrue
\ifhyper\usepackage{hyperref}\fi
\usepackage{microtype}
\DisableLigatures{encoding = *, family = *}
\usepackage{amsthm}
\usepackage{thmtools,thm-restate}
\newenvironment{proofof}[1]{\begin{trivlist} \item {\bf Proof
#1:~~}}
  {\qed\end{trivlist}}

\declaretheorem[numberwithin=section]{theorem}
\declaretheorem[numberlike=theorem]{lemma}
\declaretheorem[numberlike=theorem]{claim}

\declaretheorem[numberlike=theorem]{definition}
\declaretheorem[numberlike=theorem]{fact}
\declaretheorem[numberlike=theorem]{corollary}

\declaretheorem[numbered=no,name=Theorem]{thma}
\newcommand{\eps}{\epsilon}

\newcommand{\etal}{{\em et al.\ }}

\newcommand{\Stab}{{\mathbb S}}

\newcommand{\Var}{\operatorname{Var}}

\newcommand{\sgn}{\mathrm{sign}}

\newcommand{\ignore}[1]{}

\newcommand{\cref}[1]{Corollary~\ref{cor:#1}}

\newcommand{\bits}{\{-1,1\}}

\newcommand{\R}{{\mathbb{R}}}
\newcommand{\N}{{\bf N}}

\newcommand{\E}{\operatorname{{\bf E}}}

\newcommand{\Maj}{\mathrm{Maj}}
\newcommand{\Inf}{\mathrm{Inf}}

\newcommand{\wh}{\widehat}

\newcommand{\toD}{\stackrel{d}{\to}}
\newcommand{\pdiff}[2]{\frac{\partial #1}{\partial #2}}
\newcommand{\pdiffII}[3]{\ifthenelse{\equal{#2}{#3}}
{\frac{\partial^2 #1}{\partial #2^2}}
{\frac{\partial^2 #1}{\partial #2 \partial #3}}
}

\renewcommand{\Pr}{\operatorname{{\bf Pr}}}

\makeatletter
\renewcommand{\subsection}{\@startsection{subsection}{2}{0pt}{-6pt}{-5pt}{\normalsize\bf}}
\renewcommand{\subsubsection}{\@startsection{subsubsection}{3}{0pt}{-12pt}{-5pt}{\normalsize\bf}}
\makeatother

\date{}

\begin{document}

\title{
Majority is Stablest : Discrete and SoS}

\author{Anindya De\thanks{{\tt anindya@cs.berkeley.edu}.  Research supported by Umesh Vazirani's Templeton Foundation Grant 21674.}\\ 
University of California, Berkeley\\
\and
Elchanan Mossel\thanks{{\tt mossel@stat.berkeley.edu}. Research supported by NSF award DMS-1106999
and DOD ONR grant N000141110140  }\\ 
University of California, Berkeley\\ 
\and
Joe Neeman\thanks{{\tt jneeman@stat.berkeley.edu}. Research supported by NSF award DMS-1106999
and DOD ONR grant N000141110140 }\\ 
University of California, Berkeley\\
}

\maketitle

\setcounter{page}{0}

\thispagestyle{empty}

~
\vskip -.5in
~

\begin{abstract}
The Majority is Stablest Theorem has numerous applications in hardness of approximation and social choice theory. 
We give a new proof of the Majority is Stablest Theorem by induction on the dimension of 
the discrete cube. Unlike the previous proof, it uses neither the "invariance principle"  nor Borell's result in Gaussian space. 
 The new proof is general enough to include all previous variants of majority is stablest such as "it ain't over until it's over" 
 and "Majority is most predictable". Moreover, the new proof allows us to derive a proof of Majority is Stablest in a constant level of the
 Sum of Squares hierarchy. This implies in particular that Khot-Vishnoi instance of Max-Cut does not provide a gap instance for the Lasserre hierarchy.

\end{abstract}

\newpage

\renewcommand{\P}{{\bf P}}

\section{Introduction}

The proof of the Majority is Stablest Theorem~\cite{MOO:10} affirmed a conjecture in hardness of 
approximation~\cite{KKMO07} and in social choice~\cite{Kalai:02}. The result has been since extensively used in the two areas. 
One of the surprising features of the proof of~\cite{MOO:10} is the crucial use of deep results in Gaussian analysis~\cite{Borell:85} and an ``Invariance Principle" that allows to deduce the discrete result from the Gaussian one. 

Since the statement of the Majority is Stablest Theorem~\cite{MOO:10} deals with functions on the discrete cube, it is natural to ask (as many have) if there is a ``discrete proof" of the statement that Majority is Stablest. In this paper we answer this question affirmatively and provide a short general proof of the Majority is Stablest Theorem. The proof does not rely on Borell's result, nor does it rely on the ``Invariance Principle". 

We also show how the new proof can be transformed into a "Sum of Squares" proof of the Majority is Stablest Theorem, thus showing that Khot-Vishnoi instance of Max-cut~\cite{KV05} does not provide an integrality gap instance for Max-cut in the Lasserre hierarchy. 
\subsection{Functions with low influence variables}
In discrete Fourier analysis, special attention is devoted to functions $f : \bits^n
\to \{0,1\}$ with low influences. The $i$th influence of $f$ is defined by
\begin{equation} \label{eqn:classical-influence}
\Inf_i(f) = \P[f(x_1, \dots, x_n) \neq f(x_1, \dots, x_{i-1},
-x_i, x_{i+1}, \dots, x_n)], 
\end{equation}
where $\P$ denotes the uniform distribution on the discrete cube. 

Functions with low influences have played a crucial role in the development of the theory of discrete Fourier analysis. 
Starting with 
Kahn, Kalai, and Linial~\cite{KKL:88,Talagrand:94,FriedgutKalai:96}, the use of hyper-contractive estimates applied to low influence variables is one of the main techniques in discrete Fourier analysis. 

Of particular interest are functions all of whose influences are low. 
The work of Friedgut and Kalai~\cite{FriedgutKalai:96} shows that low influence functions have sharp thresholds. 
Central work in theoretical computer science~\cite{Bourgain:02,DinurSafra:05,SamorodnitskyTrevisan:06} pointed to the importance of low influence functions, including in the context of the ``Unique Games Conjecture"~\cite{Khot:02,KKMO07}. Such functions have also attracted much interest in the theory of social choice, see e.g.~\cite{FelsenthalMachover:98,Kalai:04}. 

In the context of voting it is natural to exclude voting schemes that give individual voters too much power. 
The same is true in the theory of hardness of approximation where a central concept is to distinguish between functions that really depend on many variables versus those who have a strong dependency on a small number of variables, see e.g.~\cite{Hastad:97,Khot:02,DinurSafra:05}. 

The Majority is Stablest theorem has been crucial in developments in both hardness of approximation and 
the theory of social choice. The theorem considers the correlation between $f(x)$ and $f(y)$ where $x,y \in \bits^n$ are $\rho$-correlated vectors with $\rho>0$. Assuming $\E[f] = 1/2$, the function that maximizes $\E[f(x)f(y)+ (1-f(x)) (1-f(y))]$ is a dictator function. The majority is stablest theorem states that for functions with low influences the value of $\E[f(x) f(y) + (1-f(x)) (1-f(y))]$ cannot be much larger than the corresponding value for the majority function. More formally,

\begin{definition} 
For $\rho \in (-1,1)$, the \emph{noise stability} of $f : \bits^n
\to \R$ at $\rho$ is defined to be
\[
\mathrm{Stab}_{\rho}(f) := \E[f(x)f(y)+ (1-f(x)) (1-f(y)) ], 
\]
when $(x,y) \in \bits^n \times
\bits^n$ is chosen so that $(x_i,y_i) \in \bits^2$ are independent
random variables with $\E[x_i] = \E[y_i] = 0$ and $\E[x_i y_i] =
\rho$.
\end{definition}

\begin{theorem}\emph{\textbf{``Majority Is Stablest''~\cite{MOO:10}}} \label{thm:MIST_intro}
Let $0 \leq \rho \leq 1$ and $\eps > 0$ be given.  Then there
exists $\tau > 0$ such that if $f : \bits^n \to [0,1]$ satisfies
$\E[f] = 1/2$ and $\Inf_i(f) \leq \tau$ for all $i$, then
\[
\mathrm{Stab}_{\rho}(f) \leq 1 - {\textstyle \frac{\arccos \rho}{\pi}}  + \eps.
\]
\end{theorem}
By Sheppard's Formula~\cite{Sheppard:99}, the quantity
$1- \frac{\arccos \rho}{\pi}$ is precisely $\lim_{n \to \infty}
\mathrm{Stab}_{\rho}(\Maj_n)$, where 
\[\Maj_n(x_1,\ldots,x_n) = \sgn(\sum_{i=1}^n x_i),
\]

We also remark here that Theorem~\ref{thm:MIST_intro} readily generalizes 
to the case when $\E[f] \not =1/2$ with the right hand side replaced by the corresponding quantity 
for a shifted majority with the same measure. 
This statement of Majority is Stablest was conjectured in~\cite{KKMO07} in 
the context of  hardness of approximation
for Max-Cut. By assuming that Theorem~\ref{thm:MIST_intro} holds, the authors
showed that it is `Unique Games-hard''  to approximate
the maximum cut in graphs to within a factor greater than
$.87856\!\dots$. This result is optimal, since the efficient algorithm of Goemans and
Williamson~\cite{GoemansWilliamson:95} is guaranteed to find
partitions that cut a $.87856\!\dots$ fraction of the maximum.
A closely related conjecture (for $\rho=-1/3$) was made 
by Kalai in the context of Arrow's Impossibility Theorem~\cite{Kalai:02}. 
The results of~\cite{MOO:10} imply Kalai's conjecture and show that Majority minimizes 
the probability of Arrow's paradox in ranking $3$ alternatives using a balanced ranking function $f$. 
See ~\cite{Kalai:02,MOO:10} for more details.

The {\em statement} of the theorem deals with Boolean functions, yet {\em the proof} of~\cite{MOO:10} crucially relies on Gaussian analysis as (a) it uses a deep result (with a hard proof) of Borell~\cite{Borell:85} on
noise stability in Gaussian space and 
(b) it uses the invariance principle developed in~\cite{MOO:10} that allows to deduce discrete statements from Gaussian statements. This raises the following natural (informal) question:

\paragraph{Question:} Is there a "discrete" proof of Majority is Stablest?\\

In other words, does there exist a proof of Majority is Stablest not using Borell's result? or any other result in Gaussian space? 
We note that almost all prior results in discrete Fourier analysis do not use Gaussian results. In particular, the classical hyper-contractive estimates~\cite{Bonami:70,Bec75} are proved by induction on dimension in the discrete cube. Moreover, most of the results in the area starting from KKL including~\cite{KKL:88,Talagrand:94,FriedgutKalai:96,Bourgain:02} do not require sophisticated results in Gaussian geometry. 

In our main result we provide a positive answer to the question above. Informally we show that

\paragraph{Main Result:} There is a proof of Majority is Stablest by induction on dimension.\\ 

Our proof is short and elegant and involves only elementary calculus and hyper-contractivity. The main difficulty in the proof is finding the right statement to prove by induction. The induction statement involves a certain function $J$, which was recently used in the derivation of a robust version of Borell's result and Majority is Stablest~\cite{MosselNeeman:12b} using Gaussian techniques and the invariance principle. 

In a way, our results here are an analogue of Bobkov's famous inequality in the discrete cube~\cite{Bobkov:97b}. Bobkov proved by induction a discrete functional inequality that at the limit becomes the Gaussian isoperimetric inequality.  Moreover, Bobkov's functional is crucial for the semi-group proof of the Gaussian isoperimetric inequality. In~\cite{MosselNeeman:12b} a functional version of Borell's result is defined and proved using the "semi-group" method in Gaussian space. Here we prove a discrete version of the same functional inequality. 

It is well known that the Majority is Stablest Theorem implies Borell's result. Here we show how this can be done by elementary methods only (our proof of Borell's result does not even require hyper-contractivity!). Our proof of Borell's result joins a number of recent proof of the result including using spherical symmetrization, see e.g.~\cite{IsakssonMossel:12}, sub-additivity~\cite{KindlerODonnell:12} and a semi-group proof~\cite{MosselNeeman:12b}. It is the simplest proof of Borell's result using elementary 
 arguments only (\cite{IsakssonMossel:12} uses sophisticated spherical re-arrangement inequalities, \cite{KindlerODonnell:12} only works for sets of measure $1/2$ and certain noise values and~\cite{MosselNeeman:12b} requires basic facts on the Orenstein-Uhlenbeck process). 

Since it was proved, Theorem~\ref{thm:MIST_intro} was generalized a number of times including in~\cite{DiMoRe:06,Mossel:10}. 
The results and their generalization have been used numerous times in hardness of approximation and social choice including in~\cite{Austrin07, OW:08, Rag:08, Mossel-arrow:12, FKN:08} . Our simple proof extends to cover all of the generalization above. It also enables to prove an SoS version of the statement of Majority is Stablest, thus answering the main open problem of \cite{OZ:12}  as we discuss next.

\subsection{Sum of Squares proof system} We now discuss an application of our new proof of Majority is Stablest to hardness of approximation. To discuss the application, we will first need to introduce the ``Sum of Squares" (SoS) proof system. In a nutshell, the SoS proof system is an algebraic proof system (introduced by Grigoriev and Vorobjov~\cite{GV01}) where constraints are encoded by polynomial (in)equalities and the deduction rules are specified by a restricted class of polynomial operations. Viewing this proof system as a refutation system for polynomial inequalities, the goal is to show that the given system of constraints is infeasible by using the allowed polynomial operations to arrive at a polynomial constraint which is ``obviously" infeasible. 

Without further ado, we introduce the following notation: let $X = (X_1, \ldots, X_n)$ be a sequence of variables and let $\mathbb{R}[X]$ be the ring of polynomials on $X$. Let $A = \{p_1 \ge 0, \ldots, p_m \ge 0\}$ be a set of constraints (on $X=(X_1, \ldots , X_n)$). Also, let $\mathbb{M}[X] \subset \mathbb{R}[X]$ be the set of polynomials which can be expressed as sums-of-squares. In other words, $q \in  \mathbb{M}[X] $ if and only if $q= r_1^2 + \ldots +r_{\ell}^2$ where $r_1, \ldots, r_{\ell} \in \mathbb{R}[X]$. For $S \subseteq [m]$, we use $p_S$ to denote $\prod_{i \in S} p_i$ with $p_{\emptyset} =1$. Now, suppose that for all $S \subseteq [m]$
there exists $q_S \in  \mathbb{M}[X] $ such that
$$
-1 = \sum_{S \subseteq [m]} p_S \cdot q_S
$$
Then, it is clear that the constraint set $A$ is infeasible over $\mathbb{R}^n$. The surprisingly powerful theorem of Stengle~\cite{Ste:74} (and earlier shown by Krivine~\cite{Kri64}) shows that whenever $A$ is infeasible, such a certificate of infeasibility always exists. This theorem is known as Stengle's Positivstellensatz. In fact, provided a certain compactness condition holds, the certificate of infeasibility (i.e. the set $\{q_S : S \subseteq [m]\}$) can always be assumed to have $q_S =0$ for $|S|>1$; this is due to Putinar~\cite{Putinar:93}. 

While these results were well-known in the algebraic geometry community and are intimately tied to Hilbert's seventeenth problem~\cite{Hil:88}, the interest in the theoretical computer science community is relatively new. The first  to view Stengle's positivstellensatz as a proof system for refutation were Grigoriev and Vorobjov~\cite{GV01} (It should be mentioned that an earlier paper \cite{LMR:96} also considered the proof theoretic aspects of Positivstellensatz but no attempt was made to quantify the complexity of such proofs). From the point of view of complexity theory, it is interesting to consider restricted proof systems where one only looks at proofs of refutation where $\max \deg (p_S \cdot q_S) \le d$. We refer to this  as the degree-$d$ SoS hierarchy. This is essentially the dual of $d/2$-level of the Lasserre hierarchy~\cite{Las:01}. 

The reason to consider the degree-$d$ SoS hierarchy is that while one loses completeness (i.e. infeasible constraint sets $A$ may not have a proof of refutation in the degree-$d$ SoS hierarchy for a fixed $d$), the degree-$d$ SoS hierarchy is \emph{effective} in the sense that if the set $A$ has a proof of infeasibility of degree $d$, then it can be found in time $O(m \cdot n^{O(d)})$ using semidefinite programming (see Parrillo~\cite{Par:00} and Lasserre~\cite{Las:01}). It should be mentioned that the so called \emph{Lasserre hierarchy}~\cite{Las:01}  and the SoS hierarchy are essentially duals of each other. So, for the subsequent discussion, whenever we use the term Lasserre hierarchy, we mean the Lasserre / SoS hierarchy.

Given that the degree-$d$ SoS hierarchy is automatizable, several researchers tried to understand the limitations of its power. Grigoriev~\cite{Gri01} showed linear lower bounds for proofs of refutation of Tseitin tautologies and the $mod \ 2$ principle.  The latter result was essentially rediscovered by Schoenebeck in the Lasserre world independently~\cite{Schoe:08}. 


\emph{Applications to  hardness of approximation:} While the results of Parillo~\cite{Par:00} and Lasserre~\cite{Las:01} have been known for more than a decade, there were only a few works  in the theoretical computer science community which harnessed the algorithmic power of \cite{Par:00, Las:01} (see \cite{BRS:11, CS08}). In fact, for the results which did use Lasserre hierarchy, it was not clear if the full power of Lasserre hierarchy was required, or whether weaker hierarchies, like the one of Lovasz and Schrijver, would suffice.

However, in a recent exciting paper, Barak \etal \cite{BBHKSZ:12} used the degree-$8$ SoS hierarchy to refute the known integrality gap instances for Unique Games~\cite{KV05, RS09, KPS10}. In other words, there are degree $8$ SoS proofs which can be used to certify that the true value of the integrality gap instances is $o(1)$. This is interesting for two reasons. The first is that even after a decade of intense investigation, these integrality gaps remained essentially the only evidence towards the truth of the Unique Games Conjecture (UGC). Thus, the SoS hierarchy discredits these instances as evidence towards the truth of the UGC. The second reason is that these integrality gaps were known to survive $\Omega( (\log \log n)^{1/4})$ rounds of weaker hierarchies like ``SDP  + Sherali Adams" \cite{RS09} or ``Approximate Lasserre"~\cite{KS09}. Thus, this showed a big gap between the Lasserre/SoS hierarchy and the weaker hierarchies like ``SDP+Sherali Adams" or ``Approximate Lasserre". 

We now mention the main idea behind showing that degree-$8$ SoS hierarchy refutes the known integrality gap instances for Unique Games~\cite{KV05, RS09, KPS10}.
Analyzing the true optimum of these instances  uses tools from analysis like hypercontractivity~\cite{Bonami:70, Bec75}, the KKL theorem~\cite{KKL:88} etc.  Hence, to show that degree $d$-SoS hierarchy can refute these instances, one essentially needs to prove SoS versions of these statements in the degree-$d$ SoS hierarchy. Note that so far we have only viewed the SoS as a refutation system, but in fact, as we will see a little later, there is an easy extension of the earlier definition, which formalizes the notion of proving a statement in the degree $d$-SoS hierarchy. In particular, \cite{BBHKSZ:12} prove SoS versions of  results like hypercontractivity,  small-set expansion etc. 

Extending the results of  \cite{BBHKSZ:12}, O'Donnell and Zhou~\cite{OZ:12} analyze the  problems ``upward" of unique games like \textsf{MAX-CUT} and
\textsf{BALANCED-SEPARATOR}. In particular, \cite{OZ:12} refutes the integrality gap instances of balanced separator from \cite{DKSV06}. Since the key to analyzing the optimum of the \textsf{BALANCED-SEPARATOR} instances in \cite{DKSV06} is the KKL theorem~\cite{KKL:88}, the authors provide a proof of the KKL theorem in the degree-$4$ SoS hierarchy. For \textsf{MAX-CUT}, their results are somewhat less powerful. Again, here they analyze the instances of \textsf{MAX-CUT} from \cite{KV05}.  More precisely, for any $\rho \in (-1,0)$, \cite{KV05} construct gap-instances of MAX-CUT where the true optimum is $\arccos \rho/ \pi + o(1)$ whereas the basic SDP-optimum is $(1-\rho)/2 + o(1)$. The key to analyzing the true optimum is the Majority is Stablest theorem of \cite{MOO:10}. Thus, to refute these instances completely i.e. show that the true optimum is $\arccos \rho/ \pi + o(1)$, the authors essentially needed to prove the Majority is Stablest theorem in some constant degree-$d$ SoS hierarchy. While the authors could not prove that, they do manage to prove the weaker ``$2/\pi$" theorem from \cite{KKMO07} in (some constant degree of) the SoS hierarchy. This implies that the SoS hierarchy can certify that the true optimum  is at most $(1/2 -\rho/\pi) -(1/2-1/\pi) \rho^3$.  They left open the problem of refuting this gap instances optimally i.e. showing that constant number of rounds of the SoS hierarchy can certify that the true optimum of these gap instances is $\arccos \rho/ \pi + o(1)$. In this paper, as the main application of the new proof of Majority is Stablest, we resolve this problem. 

It should be mentioned here that while the new proof of Majority is Stablest is more suitable for the SoS hierarchy, several powerful theorems and techniques are needed to achieve this adaptation. For example we use results from approximation theory~\cite{Lorentz:86} and a powerful matrix version of Putinar's Positivstellensatz~\cite{Lasserre:2010} to prove that a certain polynomial approximation preserves positiveness. We mention here that unlike the previous two papers \cite{BBHKSZ:12, OZ:12} connecting SoS hierarchy with hardness of approximation, we make essential use of Putinar's Positivstellensatz (i.e. essentially the completeness of the SoS hierarchy). The following is the main theorem concerning the power of SoS hierarchy on \textsf{MAX-CUT} instances.
\begin{thma}\emph{\textbf{SoS-version of MAX-CUT}}
 For every $\delta \in (0,1)$ and $\rho \in (-1,0)$, $\exists d = d(\delta, \rho)$ such that the degree-$d$ SoS hierarchy can certify that the \textsf{MAX-CUT} instances from \cite{KV05} with noise $\rho$ have true optimum less than $\arccos \rho/ \pi + \delta$. 
\end{thma}

As the key intermediate theorem, we establish a SoS version of the well-known version of the Majority is Stablest theorem (for $\rho \in (-1,0)$) which is stated next informally. 
\begin{thma}\emph{\textbf{SoS-version of Majority is Stablest}}
For every $\delta \in (0,1)$ and $\rho \in (-1,0)$, there are constants $c= c(\delta, \rho)$ and $d=d(\delta, \rho)$ such that the following is true : Let $0 \le f(x) \le 1$ for all $x \in \{-1,1\}^n$ and $\mathrm{\max_i  Inf_i(f)} \le \tau $. There is a degree-$d$ SoS proof of the statement $\mathrm{Stab}_{\rho}(f) \ge 1 - \arccos \rho/ \pi - \delta - c \cdot \tau$.   \end{thma}
Our proof can be easily modified to give the analogous statement of Majority is Stablest when $\rho \in (0,1)$. Of course, we have to change the direction of the inequality as well as impose the condition that $\mathbf{E}[f]=1/2$ (this condition is not required when $\rho \in (-1,0)$).

As the reader can see, the theorem is stated very informally. This is because SoS proofs are heavy in notation and its difficult to express the precise statement without having the proper notation. However, we do  remark that SoS version of \textsf{MAX-CUT} follows easily by composing the proof of refutation of UNIQUE-GAMES instances of \cite{KV05} (done in \cite{BBHKSZ:12}) along with the \cite{KKMO07} reduction (the proof of soundness of this reduction is the step where we require the SoS version of Majority is Stablest). 
\section{Sum of Squares hierarchy}
In this section, we formally give an introduction to the   Sum of Squares (hereafter abbreviated as SoS)  hierarchy. To define the SoS hierarchy, 
let $X=(x_1, \ldots, x_n)$ and let $\mathbb{R}[X]$ be the ring of real polynomials over these variables.  We also let $\mathbb{R}_{\le d} [X]$ denote the subset of $\mathbb{R}[X]$ consisting of polynomials of total degree bounded by $d$. As before, let $\mathbb{M}[X] \subset \mathbb{R}[X]$ be the set of polynomials which can be expressed as sums-of-squares. We next define a set of constraints given as :
\begin{itemize}
\item $A_{e}  = \{p_1(X) =0 , p_2(X) = 0 , \ldots, p_m(X) =0 \}$
\item $A_{g} = \{q_1(X) \ge 0, q_2(X) \ge 0, \ldots, q_{\ell}(X) \ge 0 \}$. 
\end{itemize}
Before we go ahead, we define the set $\mathcal{M}_{n,d}[X]$ as the set of monomials over $x_1, \ldots, x_n$ of degree bounded by $d$. Also, let $\mathbb{M}_{\le d}[X]$ denote the subset of $\mathbb{M}[X]$ of polynomials of degree bounded by $d$.  
Further, if $A = A_e \cup A_g$, define $\mathbb{V}(A) = \{X :  A \textrm{ holds on } X\}$. 
We next define the (degree $d$) closure of these constraints.. 
$$
\mathcal{C}_{d}(A_e) = \{ p_s(X) \cdot p(X) :s \in [m], \ p(X) \in \mathcal{M}_{n,d}[X]  \textrm{ and } \deg(p) + \deg(p_s)  \le d \}
$$
$$
\mathcal{C}_{d}(A_g) =  \{  \prod_{i=1}^m q_i^{a_i}(X) :  a_1, \ldots, a_m \in \mathbb{Z}^+ \textrm{ and } \sum_{i=1}^m a_i \cdot \deg(q_i) \le d \}
$$  
Note that $A_g$ includes $1 \in \mathbb{R}$. 
\begin{fact}\label{fact:closure-computation} 
Given $A_e$ and $A_g$ as described above, the sets $\mathcal{C}_{d}(A_e)$ and $\mathcal{C}_d(A_g)$ can be computed in time $n^{O(d)} \cdot m$. 
\end{fact}
It is obvious that without loss of generality, we can impose the constraints : $p(X) \in \mathcal{C}_d(A_e)$, $p(X) = 0$ and for $q(X) \in \mathcal{C}_{d}(A_g)$, $q(X) \ge 0$.  Next define, 
$$
\mathcal{C}_d(A) = \{ p(X) =0 : \forall p(X) \in \mathcal{C}_{d}(A_e) \} \cup \{ q(X)  \ge 0 : \forall q(X) \in \mathcal{C}_{d}(A_g) \} 
$$
\begin{definition}\label{def:SoS}
For the constraint set $A = A_e \cup A_g$ defined above and $h(X) \in \mathbb{R}[X]$, we say $A \vdash_d h(X) \ge 0$ if and only if 
$$
h(X) = \sum_{p(X) \in \mathcal{C}_d(A_e)} \alpha_p \cdot p (X) +  \sum_{q(X) \in \mathcal{C}_d(A_g)} r_q(X) \cdot q(X)
$$
where $\alpha_p \in \mathbb{R}$, $r_q \in \mathbb{M}[X]$ and for all $q(X) \in  \mathcal{C}_d(A_g)$, $\deg(r_q) + \deg(q) \le d$. In this case, we say that $A$ degree-$d$ SoS proves $h(X) \ge 0$.

For the constraint set $A$, we say that $A \vdash_d -1 \ge 0$ if and only if there exists 
$$
-1 = \sum_{p(X) \in \mathcal{C}_d(A_e)} \alpha_p \cdot p (X) +  \sum_{q(X) \in \mathcal{C}_d(A_g)} r_q(X) \cdot q(X)
$$
with the same constraints on $\alpha_p$ and $r_q$ as above.  In this case, we say that there is a degree-$d$ SoS refutation of the constraint set $A$. 
\end{definition}
Note that we are adopting the same notation as in \cite{OZ:12}. The reason we are interested in Definition~\ref{def:SoS} is because one can efficiently decide if $A \vdash_d -1 \ge 0$ using semidefinite programming. This is because deciding if $A \vdash_d -1 \ge 0$ is equivalent to refuting  the existence of a map $\widetilde{E}  : \mathbb{R}_{\le d}[X] \rightarrow \mathbb{R} $ satisfying the following conditions (see \cite{Par:00} for more details) 
\begin{itemize}
\item $ \widetilde{E}(1) = 1$. 
\item It is a linear map i.e. for every $g,h \in \mathbb{R}_{\le d} [X]$ and $\alpha, \beta \in \mathbb{R}$,  $\widetilde{E} (\alpha g + \beta h ) = \alpha \widetilde{E}(g) + \beta \widetilde{E}(h)$. 
\item For every $h \in \mathcal{C}_d(A_e)$, $\widetilde{E}(h) =0$. 
\item For every $h \in \mathcal{C}_d(A_g)$ and $g \in \mathbb{M}_{\le d}[X]$, such that $\deg (g \cdot h) \le d$, $\widetilde{E}(g \cdot h) \ge 0$. 
\end{itemize}
A map $\widetilde{E}$ which satisfies all the above constraints is called a degree-$d$ SoS consistent map for the constraint set  $A = A_e \cup A_g$.  Lasserre~\cite{Las:01} and Parillo~\cite{Par:00} have shown that using semidefinite programming, it is possible to decide the feasibility of such a map $\widetilde{E}$ in time $m \cdot n^{O(d)}$. In fact, if there exists such a map $\widetilde{E}$, then the algorithm outputs one in the same time. It is important to mention that since $\widetilde{E}$ has an infinite domain, it is not obvious what one means by outputting the map. To see why this makes sense, note that $\widetilde{E}$ is a linear map and hence it suffices to give to specify $\widetilde{E}$ on the set $\mathcal{M}_{n,d}[X]$.  We also remark here that the notion of  finding a mapping $\tilde{E}$  is close to the viewpoint taken by Barak \etal \cite{BBHKSZ:12}. 

To get started with SoS proof systems, we state a few facts (which are very easy to prove) : 
\begin{fact}\label{fac:SoS-basic0} 
\begin{itemize}
\item If $A \vdash_d p \ge 0$ and $A' \vdash_{d'} q \ge 0$, then $A \cup A' \vdash_{\max \{d,d'\}} p+ q \ge 0$. 
\item If $A \vdash_d p \ge 0$ and $A \vdash_{d'} q \ge 0$, then $A \vdash_{d+d'} p \cdot q \ge 0$
\item If $A \vdash_d \{p_1 \ge 0 , p_2 \ge 0, \ldots , p_m \ge 0 \}$ and $\{p_1 \ge 0 , p_2 \ge 0, \ldots , p_m \ge 0 \} \vdash_{d'} q \ge 0$, $A \vdash_{d \cdot d'} q \ge 0$. 
\end{itemize}
\end{fact}
Several other SoS facts are proven in the Appendix~\ref{app:SoS}.  We suggest that the non-expert reader look at the Appendix~\ref{app:SoS} to get more comfortable with the notion of SoS proofs. 
For rest of the paper, we set the following convention for indeterminates appearing in SoS proofs :  Capital letters $X$, $Y$ and $Z$ will be used to denote a sequence of indeterminates (i.e. $X =(x_1, \ldots, x_n)$) while small letters $x$, $y$ and $z$ will be used to indicate single indeterminates. This convention is however only for indeterminates in the SoS proofs. For other variables, both capital and small letters will be used.  Also, we will consider polynomials on the indeterminates occurring in the SoS proofs. Whenever we refer to such polynomials without an explicit reference to the underlying indeterminates, the set of indeterminates will be clear from the context.  To get the reader more acquainted with the power of SoS proofs, we state the following powerful result of Putinar which we use repeatedly. 


\begin{theorem}\label{thm:Putinar}\emph{\cite{Putinar:93}}
Let $A = \{p_1(X) \ge 0, \ldots, p_m(X) \ge 0\}$  and define $\mathbb{M}(A) = \sum_{i=1}^n r_i p_i + r_0$ where $r_0, \ldots, r_m \in \mathbb{M}[X]$. Assume that $\exists q \in \mathbb{M}(A) $ such that the set $\{ X : q(X) \ge 0 \}$ is compact. 
If $p>0$ on the set $\mathbb{V}(A)$, then $p \in \mathbb{M}(A)$. \end{theorem}
As a key step in one of our proofs, we will also require a matrix version of Putinar's Positivstellensatz (see \cite{Lasserre:2010} for details). A matrix $\Gamma \in (\mathbb{R}[X])^{p \times p}$ is said to be a sum-of-squares if there exists $B \in  (\mathbb{R}[X])^{p \times q}$ (for some $q \in \mathbb{N}$) such that $B \cdot B^T = \Gamma$.

\begin{theorem}\label{thm:matrix-Putinar}Let $A = \{p_1(X) \ge 0, \ldots, p_m(X) \ge 0\}$ be satisfying the conditions in the hypothesis of  Theorem~\ref{thm:Putinar}. Let $\Gamma \in (\mathbb{R}[X])^{p \times p}$ be a symmetric matrix and $\delta>0$ be such that $\Gamma \succeq \delta I$ on the set $\mathbb{V}(A)$. Then, $\Gamma =  \Gamma_0(X) + \sum_{i=1}^m \Gamma_i(X) \cdot p_i(X)$ where $\Gamma_0 ,  \ldots, \Gamma_m$ are sum-of-squares. 
\end{theorem}



\section{Our tensorization theorem}

In this section, we will prove our main tensorization inequality on the cube.
In subsequent sections, we will use it to give new proofs
of the ``Majority is Stablest'' theorem of Mossel, O'Donnell and Oleszkiewicz~\cite{MOO:10}
and the Gaussian stability inequality of Borell~\cite{Borell:85}.
We begin by defining the following function from~\cite{MosselNeeman:12b} for every $\rho \in [-1,1]$ : $J_{\rho} : (0,1)^2 \rightarrow [0,1]$ as $$ J_{\rho}(x,y) = \Pr_{X,Y} [ X \le \Phi^{-1}(x), Y \le \Phi^{-1}(y)]$$ 
Here $X,Y$ are jointly normally distributed random variables with the covariance matrix 
$$
\mathop{Cov}(X,Y) = \left(\begin{array}{cc} 1 & \rho  \\\rho  & 1 \end{array}\right).
$$

\begin{definition}
 Let $\Omega_1$ and $\Omega_2$ be probability spaces and $\mu$ be a probability
 measure on $\Omega_1 \times \Omega_2$. We say that $\mu$ has R\'enyi correlation at most $\rho$
 if for every measurable $f: \Omega_1 \to \R$
 and $g: \Omega_2 \to \R$ with $\E_\mu f = \E_\mu g = 0$,
 \[
  \E_\mu[fg] \le \rho \sqrt{\E_\mu[f^2] \E_\mu[g^2]}.
 \]
\end{definition}

For example, suppose that $\Omega_1 = \Omega_2$ and suppose $(X, Y)$ are
generated by the following procedure: first choose $X$ according to some
distribution $\nu$. Then, with probability
$\rho$ set $Y = X$, and with probability $1-\rho$, choose
$Y$ independently from $\nu$. If $\mu$ is the distribution of $(X, Y)$, then
it is easy to check that $\mu$ has R\'enyi correlation $\rho$.

\begin{definition}\label{def:one-restriction}
 If $\Omega$ is a probability space and $f$ is a function $\Omega^n \to \R$,
 then for $X \in \Omega$, we define $f_X: \Omega^{n-1} \to \R$ by
 \[
  f_X(X_1, \dots, X_{n-1}) = f(X_1, \dots, X_{n-1}, X).
 \]
\end{definition}

\begin{definition}
For a function $f: \Omega \to \R$, define
\[
 \Delta_1(f) = \E |f - \E f|^3.
\]
For a function $f: \Omega^n \to \R$, define $\Delta_n(f)$
recursively by
\[
 \Delta_n(f) = \E_{X_n}[\Delta_{n-1}(f_{X_n})] + \Delta_1(\E[f_{X_n}|X_n]),
\]
noting that $\E[f_{X_n}|X_n]$ is a function $\Omega \to \R$.
\end{definition}

We prove the following general theorem, which we will later
use
to derive both Borell's inequality and the ``Majority is Stablest'' theorem.
\begin{theorem}\label{thm:tensorization}
For any $\epsilon > 0$ and $0 < \rho < 1$, there is $C(\rho) > 0$ such
that the following holds.
Let $\mu$ be a $\rho$-correlated measure on $\Omega_1 \times \Omega_2$ and
let $(X_i, Y_i)_{i=1}^n$ be i.i.d.\ variables with distribution $\mu$.
Then for any measurable functions
$f: \Omega_1^n \to [\epsilon, 1-\epsilon]$ and
$g: \Omega_2^n \to [\epsilon, 1-\epsilon]$,
\[
 \E J_\rho(f(X), g(Y)) \le J_\rho(\E f, \E g) + C(\rho) \epsilon^{-C(\rho)} (\Delta_n(f) + \Delta_n(g)).
\]
\end{theorem}
We note that~\cite{MosselNeeman:12b} proved that in the {\em Gaussian setup} where $f,g : \R^n \to [0,1]$ it was shown 
that $ \E J_\rho(f(X), g(Y)) \le J_\rho(\E f, \E g)$.  

\subsection{The base case}
We prove Theorem~\ref{thm:tensorization} by induction on $n$.
In this section, we will prove the base case $n=1$:
\begin{claim}\label{clm:base-case}
  For any $\epsilon > 0$ and $0 < \rho < 1$, there is a $C(\rho)$
  such that for any two random variables
  $X, Y \in [\epsilon, 1-\epsilon]$ with correlation in $[0, \rho]$,
  \[
   \E J_\rho(X, Y) \le J_\rho(\E X, \E Y) + C(\rho)\epsilon^{-C(\rho)}(\E |X - \E X|^3 + \E |Y - \E Y|^3).
  \]
\end{claim}

 The proof of Claim~\ref{clm:base-case} essentially
 follows from Taylor's theorem applied to the
 function $J$; the crucial point is that $J$ satisfies a certain differential equation.
Define the matrix $M_{\rho \sigma}(x,y)$ by
$$
M_{\rho \sigma}(x,y) =  \left(\begin{array}{cc} \frac{\partial^2 J_{\rho}(x,y)}{\partial^2 x} & \sigma \frac{\partial^2 J_{\rho}(x,y)}{\partial x \partial y} \\\sigma \frac{\partial^2 J_{\rho}(x,y)}{\partial x \partial y}  & \frac{\partial^2 J_{\rho}(x,y)}{\partial^2 x}\end{array}\right).
$$
\begin{restatable}{claim}{negsemi}
\label{clm:negative-semidefinite}
For any  $(x,y) \in (0,1)^2$ and $0 \le \sigma \le \rho$, $M_{\rho \sigma}(x,y)$ is a negative semidefinite matrix.  Likewise, if $\rho \le \sigma \le 0$, then $M_{\rho \sigma}(x,y)$ is a positive semidefinite matrix. 
\end{restatable}

\begin{restatable}{claim}{thirddiff}
\label{clm:third-derivative}
For any $-1 < \rho < 1$, there exists $C(\rho) > 0$ such that
for any $i, j \ge 0$,  $i+j=3$,
$$\left|\frac{\partial^3 J_{\rho}(x,y)}{\partial^i x \partial^j y} \right| \le C(\rho) (xy(1-x) (1-y))^{-C(\rho)} $$  
\end{restatable}

Claims~\ref{clm:negative-semidefinite} and~\ref{clm:third-derivative} follow from elementary
calculus, and we defer their proofs to the appendix (Claim~\ref{clm:negative-semidefinite} is first proved in~\cite{MosselNeeman:12b} and we include the proof here for the sake of completeness). 
  Now we will use them with Taylor's theorem to prove Claim~\ref{clm:base-case}.
 \begin{proof}[Proof of Claim~\ref{clm:base-case}]
 Fix $\epsilon > 0$ and $\rho \in (0, 1)$,
 and let $C(\epsilon)$ be large enough so that all third derivatives of $J$
 are uniformly bounded by $C(\epsilon)$ on the square $[\epsilon, 1-\epsilon]^2$
 (such a $C(\epsilon)$ exists by Claim~\ref{clm:third-derivative}).
 Taylor's theorem then implies that for any $a, b, a+x, b+y \in [\epsilon, 1-\epsilon]$,
 \begin{multline}\label{eq:taylor}
  J_\rho(a + x, b + y) \le J_\rho(a, b) + x \pdiff Jx (a, b) + y\pdiff Jy (a, b) \\ +
  \frac{1}{2} (x \ y)
  \begin{pmatrix} \pdiffII {J_\rho}xx (a, b) & \pdiffII {J_\rho}xy (a, b) \\ \pdiffII {J_\rho}xy(a, b) & \pdiffII {J_\rho}yy(a, b)\end{pmatrix}
  \begin{pmatrix}x \\ y\end{pmatrix} + C(\rho)\epsilon^{-C(\rho)} (x^3 + y^3).
 \end{multline}
 
 Now suppose that $X$ and $Y$ are random variables taking values in
 $[\epsilon, 1-\epsilon]$. If we apply~\eqref{eq:taylor} with $a = \E X$, $b = \E Y$,
 $x = X - \E X$, and $y = Y - \E Y$, and then take expectations of both sides, we obtain
 \begin{multline}\label{eq:random-taylor}
  \E J_\rho(X, Y) \le J_\rho(\E X, \E Y) + \frac{1}{2}
  \E\left[
  (\tilde X\ \tilde Y)
  \begin{pmatrix}\pdiffII {J_\rho}xx(a, b) & \pdiffII {J_\rho}xy(a, b) \\ \pdiffII {J_\rho}xy(a, b) & \pdiffII {J_\rho}yy(a, b)\end{pmatrix}
  \begin{pmatrix}\tilde X \\ \tilde Y \end{pmatrix} \right] \\
  + C(\rho)\epsilon^{-C(\rho)} (\E |\tilde X|^3 + \E |\tilde Y|^3)
 \end{multline}
 where $\tilde X = X - \E X$ and $\tilde Y = Y - \E Y$.
 Now, if $X$ and $Y$ have correlation $\sigma \in [0, \rho]$ then
 $\E \tilde X \tilde Y = \sigma \sqrt{\E \tilde X^2 \E \tilde Y^2}$, and so
 \[
  \E\left[
  (\tilde X\ \tilde Y)
  \begin{pmatrix}\pdiffII {J_\rho}xx(a, b) & \pdiffII {J_\rho}xy(a, b) \\ \pdiffII {J_\rho}xy(a, b) & \pdiffII {J_\rho}yy(a, b)\end{pmatrix}
  \begin{pmatrix}\tilde X \\ \tilde Y \end{pmatrix} \right]
  = (\sigma_X\ \sigma_Y)
  \begin{pmatrix}\pdiffII {J_\rho}xx(a, b) & \sigma \pdiffII {J_\rho}xy(a, b) \\ \sigma \pdiffII {J_\rho}xy(a, b) & \pdiffII {J_\rho}yy(a, b)\end{pmatrix}
  \begin{pmatrix}\sigma_X \\ \sigma_Y \end{pmatrix}
\]
where $\sigma_X = \sqrt{\E \tilde X^2}$ and $\sigma_Y = \sqrt{\E \tilde Y^2}$.
By Claim~\ref{clm:negative-semidefinite}.
\[
  (\sigma_X\ \sigma_Y)
  \begin{pmatrix}\pdiffII {J_\rho}xx(a, b) & \sigma \pdiffII {J_\rho}xy(a, b) \\ \sigma \pdiffII {J_\rho}xy(a, b) & \pdiffII {J_\rho}yy(a, b)\end{pmatrix}
  \begin{pmatrix}\sigma_X \\ \sigma_Y \end{pmatrix}
  = 
  (\sigma_X\ \sigma_Y)
  M_{\rho\sigma}(a, b)
  \begin{pmatrix}\sigma_X \\ \sigma_Y \end{pmatrix} \le 0.
\]
Applying this to~\eqref{eq:random-taylor}, we obtain
\[
 \E J_\rho(X, Y) \le J_\rho(\E X, \E Y) + C(\rho)\epsilon^{-C(\rho)} (\E |\tilde X|^3 + \E |\tilde Y|^3)
\]
\end{proof}

\subsection{The inductive step}
Next, we prove Theorem~\ref{thm:tensorization} by induction.
\begin{proof}[Proof of Theorem~\ref{thm:tensorization}]
 Suppose that the Theorem holds with $n$ replaced by $n-1$.
 Consider $f: \Omega_1^n \to [\epsilon, 1-\epsilon]$ and
 $g: \Omega_2^n \to [\epsilon, 1-\epsilon]$.
 
 Conditioning on $(X_n, Y_n)$
 and writing $\tilde X = (X_1, \dots, X_{n-1})$, $\tilde Y = (Y_1, \dots, Y_{n-1})$, we have
 \[
  \E J_\rho(f(X), g(Y)) = \E_{X_n,Y_n} \E_{\tilde X,\tilde Y} J_\rho(f_{X_n}(\tilde X), g_{Y_n}(\tilde Y)).
 \]
 Applying the inductive hypothesis for $n-1$ conditionally on $X_n$ and $Y_n$,
 \begin{multline}\label{eq:induction-1}
  \E_{\tilde X, \tilde Y} J_\rho(f_{X_n}(\tilde X), g_{Y_n}(\tilde Y)) \\
  \le J_\rho(\E[f_{X_n}| X_n], \E[g_{Y_n}|Y_n]) + C(\rho)\epsilon^{-C(\rho)} (\Delta_{n-1}(f_{X_n}) + \Delta_{n-1}(f_{Y_n})).
 \end{multline}
 On the other hand, the base case for $n=1$ implies that
 \begin{multline}\label{eq:induction-2}
  \E_{X_n,Y_n} J_\rho(\E[f_{X_n}| X_n], \E[g_{Y_n}|Y_n]) \\
  \le J_\rho(\E f, \E g)
  + C(\rho)\epsilon^{-C(\rho)} (\Delta_1 (\E[f_{X_n}| X_n]) + \Delta_1 (\E[g_{Y_n}|Y_n])).
 \end{multline}
 Taking the expectation of~\eqref{eq:induction-1} and combining it with~\eqref{eq:induction-2},
 we obtain
 \begin{multline*}
  \E J_\rho(f(X), g(Y)) \le J_\rho(\E f, \E g) \\
  + C(\rho)\epsilon^{-C(\rho)} \big(\E_{X_n} \Delta_{n-1}(f_{X_n}) + \E_{Y_n} \Delta_{n-1}(f_{Y_n}) \\
  + \Delta_1 (\E[f_{X_n}| X_n]) + \Delta_1 (\E[g_{Y_n}|Y_n])\big).
 \end{multline*}
 Finally, note that the definition of $\Delta_n$ implies that the right-hand side above is just
 \[
  J_\rho(\E f, \E g) + C(\rho)\epsilon^{-C(\rho)}\big(
  \Delta_n (f) + \Delta_n(g)
  \big).
 \]
\end{proof}

\section{Borell's inequality}
The most interesting special case of Theorem~\ref{thm:tensorization}
is when $\Omega_1 = \Omega_2 = \{-1, 1\}$ and the distributions of $X_i$, $Y_i$
satisfy $\E X_i = \E Y_i = 0$, $\E X_i Y_i = \rho$.
In this section and the next, we will focus on this special case.
First, let us recall the functional version of Borell's inequality
that was given in~\cite{MosselNeeman:12b}.

\begin{theorem}\label{thm:borell}
 Suppose that $G_1$ and $G_2$ are Gaussian vectors with joint distribution
 \[
  (G_1, G_2) \sim \mathcal{N}\left(0, \begin{pmatrix} I_d & \rho I_d \\ \rho I_d & I_d\end{pmatrix}\right).
 \]
 For any measurable $f_1, f_2: \R^d \to [0, 1]$,
 \[
  \E J(f_1(G_1), f_2(G_2)) \le J(\E f_1, \E f_2).
 \]
\end{theorem}

We will prove Theorem~\ref{thm:borell} using Theorem~\ref{thm:tensorization}
and a crude bound on $\Delta_n(f)$ (in the next section, we will need
a much better bound on $\Delta_n(f)$ to prove that ``Majority is Stablest'').

\begin{claim}\label{clm:easy-delta-bound}
 For $X \in \{-1, 1\}^n$, define
 \[
  X^{-i} = (X_1, \dots, X_{i-1}, -X_i, X_{i+1}, \dots, X_n).
 \]
 Then
 \[
  \Delta_n(f) \le \sum_{i=1}^n \E |f(X) - f(X^{-i})|^3.
 \]
\end{claim}

\begin{proof}
 The proof is by induction: the base case is trivial, while the inductive step
 follows by Jensen's inequality:
 \begin{align*}
  \Delta_n(f)
  &= \E_{X_n}[\Delta_{n-1}(f_{X_n})] + \E_{X_n} |\E[f_{X_n} | X_n] - \E f|^3 \\
  &\le  \sum_{i=1}^{n-1} \E |f(X) - f(X^{-i})|^3 + \E_{X_n} |\E[f_{X_n} | X_n] - \E [f_{-X_n} | X_n]|^3 \\
  &\le  \sum_{i=1}^{n-1} \E |f(X) - f(X^{-i})|^3 + \E |f(X) - f(X^{-n})|^3.
 \end{align*}
\end{proof}

\begin{proof}[Proof of Theorem~\ref{thm:borell}]
 Let $n = md$ and, for each $i = 1, \dots, d$, define
 \begin{align*}
  G_{1,n} &= \frac{1}{\sqrt m} \left(
  \sum_{i=1}^m X_i,
  \sum_{i=m+1}^{2m} X_i,
  \dots,
  \sum_{i=(m-1)d + 1}^{md} X_i
  \right).
 \end{align*}
 Define $G_{2,n}$ similarly by with $Y$ instead of $X$. By the multivariate
 central limit theorem, $(G_{1,n}, G_{2,n}) \toD (G_1, G_2)$ as $m \to \infty$.
 
 Suppose first that $f_1$ and $f_1$ are $L$-Lipschitz functions taking values
 in $[\epsilon, 1-\epsilon]$, and define $g_1, g_2: \{-1, 1\}^n$ by
 $g_i(X) = f_i(G_{1,n})$. By Theorem~\ref{thm:tensorization},
 \begin{equation}\label{eq:borell-proof-1}
  \E J(g_1(X), g_2(Y)) \le J(\E g_1, \E g_2) + C(\epsilon) (\Delta_n(g_1) + \Delta_n(g_2)).
 \end{equation}
 Since $f_i$ is $L$-Lipschitz,
 \[
  |g_i(X) - g_i(X^{-j})| \le \frac{2L}{\sqrt m}
 \]
 for every $j$, and so Claim~\ref{clm:easy-delta-bound} implies that
 \[
  \Delta_n(g_i) \le \frac{8 L^3 n}{m^{3/2}} = \frac{8 L^3 d}{\sqrt m}.
 \]
 Applying this to~\eqref{eq:borell-proof-1},
 \[
  \E J(g_1(X), g_2(Y)) \le J(\E g_1, \E g_2) + C(\epsilon)\frac{16 L^3 d C(\epsilon)}{\sqrt m}
 \]
 and so the definition of $g_i$ implies
 \[
  \E J(f_1(G_{1,n}), f_2(G_{2,n})) \le J(\E f_1(G_{1,n}), \E f_2(G_{2,n})) + C(\epsilon)\frac{16 L^3 d C(\epsilon)}{\sqrt m}
 \]
 Taking $m \to \infty$, the central limit theorem implies that
 \begin{equation}\label{eq:borell-lipschitz}
  \E J(f_1(G_{1}), f_2(G_{2})) \le J(\E f_1(G_{1}), \E f_2(G_{2})).
 \end{equation}

 This establishes the theorem for functions $f_1$ and $f_2$ which are Lipschitz and take values
 in $[\epsilon, 1-\epsilon]$. But any measurable $f_1, f_2: \R^d \to [0, 1]$ can be approximated
 (say in $L^p(\R^d, \gamma_d)$) by Lipschitz functions with values in $[\epsilon, 1-\epsilon]$.
 Since neither the Lipschitz constant nor $\epsilon$ appears in~\eqref{eq:borell-lipschitz},
 the general statement of the theorem follows from the dominated convergence theorem.
\end{proof}

\section{Majority is stablest}

By giving a bound on $\Delta_n(f)$ that is better than Claim~\ref{clm:easy-delta-bound},
we can derive the ``Majority is Stablest'' theorem from Theorem~\ref{thm:tensorization}.
Indeed, we can express $\Delta_n(f)$ in terms of the Fourier coefficients
of $f$, and we can bound $\Delta_n(f)$ in terms of the max influence of $f$.
For this, we will introduce some very basic Fourier analytic preliminaries below. 
\subsection*{Fourier analysis:} We start by defining the ``character" functions i.e. for every $S \subseteq [n]$, define $\chi_S(x) : \{-1,1\}^n \rightarrow \mathbb{R}$ as
$\chi_S(x) = \prod_{i \in S} x_i$. Now, every function $f : \{-1,1\}^n \rightarrow \mathbb{R}$ can be expressed as 
$$
f(x) = \sum_{S \subseteq [n]} \wh{f}(S) \chi_S(x) \quad \quad  \wh{f}(S) = \mathop{\mathbf{E}}_{x \in \{-1,1\}^n} [f(x) \cdot \chi_S(x)] 
$$
The coefficients $\wh{f}(S)$ are called the Fourier coefficients of $f$ and the expansion of $f$ in terms of $\wh{f}(S)$ is called the Fourier expansion of $f$. It is easy to show that
$\sum_{S \subseteq[n]} \wh{f}^2(S) =  \mathop{\mathbf{E}}_{x \in \{-1,1\}^n} [f^2(x)]$. This is known in literature as Parseval's identity. Similarly, for any $\rho \in [-1,1]$ and $x \in \{-1,1\}^n$, we define $y \sim_{\rho} x$ as the distribution over $\{-1,1\}^n$ where every bit of $y$ is independent and $\mathbf{E}[x_i y_i ] =\rho$. This immediately  lets us define the noise operator $T_{\rho}$ as follows : For any function $f : \{-1,1\}^n \rightarrow \mathbb{R}$, $T_{\rho} f(x) = \mathbf{E}_{y \sim_{\rho} x} [f(y)]$. The effect of the noise operator $T_{\rho}$ is particularly simple to describe on the fourier spectrum. $\wh{T_{\rho} f}(S) = \rho^{|S|} \wh{f}(S)$. The reader is referred to the excellent set of lecture notes by Ryan O'Donnell \cite{ODonnell:BFA} for an extensive reference on this topic. 

It is also important to remark here that while we prove the ``Majority is Stablest'' theorem for the hypercube with the uniform measure, one can easily derive analogues of this theorem for more general product spaces by extending our machinery. Instead of using the fourier expansion of the function, one has to use the Efron-Stein decomposition (see the lecture notes by Mossel~\cite{Mossel:05a} for an extensive reference on the Efron-Stein decomposition). All the statements that we prove here have analogues in the Efron-Stein world. We leave it to the expert reader to fill in the details. 

We start by extending  the notation of Definition~\ref{def:one-restriction}:

\begin{definition}
For disjoint sets $S, T \subset [n]$, and elements $x \in \{-1, 1\}^S, y \in \{-1, 1\}^T$,
we write $x\cdot y$ for their concatenation in
$\{-1, 1\}^{S \cup T}$.

For a function $f: \{-1, 1\}^n \to \R$, a set $S \subset [n]$,
and an element $x \in \{-1, 1\}^S$, we define
$f_x: \{-1, 1\}^{[n] \setminus S} \to \R$ by
$f_x(y) = f(x \cdot y)$.
\end{definition}

Our first observation is that $\Delta_n$ of $f$ can be written in terms of
Fourier coefficients of random restrictions of $f$.

\begin{claim}\label{clm:delta-fourier}
If $S_i = \{i+1, \dots, n\}$, then
 \[
  \Delta_n(f) = \sum_{i=1}^{n} \E_{X \in \{-1, 1\}^{S_i}} |\wh{f_X}(i)|^3.
 \]
\end{claim}

\begin{proof}
 The proof is by induction. The base case is just the fact that
 for a function $f: \{-1, 1\} \to \R$,
 \[
  |\wh{f}(1)|^3 = \Big|\frac{f(1) - f(-1)}{2}\Big|^3 = \E |f - \E f|^3.
 \]
 For the inductive step, we have
 \begin{align*}
  \Delta_n(f)
  &= \E_{X_n}[\Delta_{n-1}(f_{X_n})] + \Delta_1(\E [f_{X_n} | X_n]) \\
  &= \E_{X_n} \left[\sum_{i=1}^{n-1} \E_{X_{i+1}, \dots, X_{n-1}} |\wh{f_X} (i)|^3\right] + |\wh f(n)|^3 \\
  &= \sum_{i=1}^n \E_{X \in \{-1, 1\}^{S_i}} |\wh{f_X}(i)|^3.
 \end{align*}
\end{proof}

In order to control the Fourier coefficients of restrictions of $f$, we can
write them in terms of the Fourier coefficients of $f$:
\begin{claim}\label{clm:expansion-of-restriction}
For any disjoint $S$ and $U$ and any $x \in \{-1, 1\}^S$,
\[
 \wh{f_{x}}(U) =
 \sum_{T \subset S} \chi_T(x) \wh{f}(T \cup U).
\]
\end{claim}
\begin{proof}
Fix $S$ and $x$.
 Let $g: \{-1, 1\}^n \to \R$ be the function such that $g(y) = 1$ when $y_i = x_i$ for all $i \in S$, and $g(y) = 0$ otherwise.
 It is easy to check that the Fourier expansion of $g$ is
 \[
  g(y) = 2^{-|S|} \sum_{T \subset S} \chi_T(x) \chi_T(y).
 \]
 Then
 \begin{align*}
  \wh{f_{x}}(U)
  &= \E_{X_{[n] \setminus S}} f_{x}(X_{[n] \setminus S}) \chi_U(X_{[n] \setminus S}) \\
  &= 2^{|S|} \E_{X} f(X) g(X) \chi_U(X) \\
  &= \E_{X} f(X) \sum_{T \subset S} \chi_T(x) \chi_T(X) \chi_U(X) \\
  &= \sum_{T \subset S} \chi_T(x) \wh{f}(T \cup U).
 \end{align*}
\end{proof}

In particular, the identity of Claim~\ref{clm:expansion-of-restriction} allows
us to compute second moments of $\wh f_X$:

\begin{claim}\label{clm:squared-restrictions}
 For any function $f: \{-1, 1\}^n \to \R$, any $x \in \{-1, 1\}^S$ and any $i \in U \subset [n]$,
 \[
  \E_{X \in \{-1, 1\}^S} |\wh{f_{X}}(U)|^2 \le \Inf_i(f).
 \]
 Moreover, if $S_i = \{i+1, \dots, n\}$ then
 \[
  \sum_{i=1}^n \E_{X \in \{-1, 1\}^{S_i}} |\wh{f_X}(i)|^2 = \Var(f).
 \]
\end{claim}

\begin{proof}
 In view of Claim~\ref{clm:expansion-of-restriction}, we can write
 \[
  |\wh{f_{X_S}}(U)|^2
  = \sum_{T,T'\subset S} \chi_T(X_S) \chi_{T'}(X_S) \wh f(T \cup U) \wh f(T' \cup U).
 \]
 When we take the expectation with respect to $X_S$,
 $\E \chi_T(X_S) \chi_{T'}(X_S) = \delta_{T,T'}$ and so the cross-terms vanish:
 \begin{equation}\label{eq:one-squared-restriction}
  \E_{X_S} |\wh{f_{X_S}}(U)|^2
  = \sum_{T\subset S} \wh f^2(T \cup U).
 \end{equation}
 Since $\Inf_i(f) = \sum_{T \ni i} \wh f^2(T)$, the first part of the claim follows.
 
 For the second part,
 \begin{equation}\label{eq:squared-restrictions-1}
  \sum_{i=1}^n \E_{X \in \{-1, 1\}^{S_i}} |\wh{f_X}(i)|^2
  = \sum_{i=1}^n \sum_{T \subset S_i} \wh f^2(T \cup \{i\})
  = \sum_{U \subset [n], U \ne \emptyset} \wh f^2(U),
 \end{equation}
 where the last equality used the fact that every non-empty $U \subset [n]$
 can be written uniquely in the form $T \cup \{i\}$ for some $T \subset \{i+1, \dots, n\}$.
 But of course the right-hand side of~\eqref{eq:squared-restrictions-1} is just
 $\Var(f)$.
\end{proof}

Next, we will consider $\wh f_x(n-i)$ as a polynomial in $x$ and apply hypercontractivity to the
right hand side of Claim~\ref{clm:delta-fourier}. First, note that $T_\sigma$ commutes (up to
a multiplicative factor) with restriction:
\begin{claim}\label{clm:T_rho-commutation}
For any $0 < \sigma < 1$, if $S, U \subset [n]$ are disjoint then,
as polynomials in $x = (x_i)_{i \in S}$,
 \[
  \wh{(T_\sigma f)_x}(U) = \sigma^{|U|} T_\sigma (\wh{f_x}(U)).
 \]
\end{claim}

\begin{proof}
By Claim~\ref{clm:expansion-of-restriction},
 \[
  \wh{f_x}(U) = \sum_{T \subset S} \chi_T(x) \wh f (T \cup U).
 \]
 Since $\wh{T_\sigma f}(T \cup U) = \sigma^{|T| + |U|} \wh f(S)$ and $T_\sigma \chi_T(x) = \sigma^{|T|}\chi_T$,
 it follows that
 \[
  \wh{(T_\sigma f)_x}(U) = \sum_{T \subset S} \sigma^{|T|+|U|} \chi_T(x) \wh f(T \cup U)
   = \sigma^{|U|} T_\sigma (\wh{f_x}(U)).
 \]
\end{proof}

Essentially, Claim~\ref{clm:T_rho-commutation} allows us to apply the Bonami-Beckner inequality
to $\wh{f_X}$: for any $\sigma < 1$, if $p = 1 + \sigma^{-2}$ then
\[
 \E_{X \in \{-1, 1\}^S} |\wh{(T_\sigma f)_X}(U)|^p \le (\E_X |\wh f_X(U)|^2)^{p/2}.
\]
By Claim~\ref{clm:squared-restrictions}, if $i \in U$ then
\[
 \E_{X \in \{-1, 1\}^S} |\wh{(T_\sigma f)_X}(U)|^p \le (\Inf_i(f))^{\frac{p-2}{2}} (\E_X |\wh f_X(U)|^2).
\]
Applying this to $S_i = \{i+1, \dots, n\}$ and $U_i = \{i\}$ and summing the result over $i = 1, \dots, n$,
we obtain
\[
 \sum_i \E_{X \in \{-1, 1\}^{S_i}} |\wh{(T_\sigma f)_X}(i)|^p \le (\max_i \Inf_i(f))^{\frac{p-2}{2}} \Var(f).
\]
Now, if $f$ takes values in $[-1, 1]$ then $\Var(f) \le 1$ and all Fourier coefficients of $f$
(and its restriction) are bounded by 1. Hence, Claim~\ref{clm:delta-fourier} implies the following:
\begin{claim}\label{clm:delta-bound}
If $1 + \sigma^{-2} \le 3$ then 
\[
 \Delta_n(T_\sigma f) \le (\max_i \Inf_i(f))^{\frac{1-\sigma^2}{2\sigma^2}}.
\]
\end{claim}

Now we are ready to prove the ``Majority is Stablest'' theorem. For this, we define $\Stab_\rho(f)$ as $$ \Stab_\rho(f)= \mathbf{E}_{x \in \{-1,1\}^n, y \sim_{\rho} x} [f(x) f(y)]$$
\begin{theorem}\label{thm:MIST}
For any $0 < \rho < 1$, there are constants $0 < c(\rho), C(\rho) < \infty$ such that for any function
$f: \{-1, 1\}^n \to [0, 1]$ with $\max_i \Inf_i(f) \le \tau$,
\[
 \Stab_\rho(f) \le J_\rho(\E f, \E f) + C(\rho) \frac{\log \log (1/\tau)}{\log (1/\tau)}.
\]
\end{theorem}

As remarked earlier,  our proof extends to the generalizations of Theorem~\ref{thm:MIST} such as those presented by \cite{DiMoRe:06, Mossel:10}. 
The extension of the proof uses the Efron-Stein decomposition instead of the Fourier decomposition. The only difference is that the hyper-contractivity parameter 
will now depend on the underlying space. See~\cite{Mossel:10} for more details.

\begin{proof}
 Suppose $f: \{-1, 1\}^n \to [\epsilon, 1-\epsilon]$ satisfies $\max_i \Inf_i(f) \le \tau$,
 and let $X, Y$ be uniformly random elements of $\{-1, 1\}^n$ with $\E X_i Y_i = \rho$.
 For sufficiently small $\eta > 0$, Claim~\ref{clm:delta-bound} implies that
 \[
  \Delta_n(T_{1-\eta} f) \le \tau^{c\eta}.
 \]
 Note that the range of $T_{1-\eta} f$ belongs to $[\epsilon, 1-\epsilon]$ because
 the range of $f$ does. Hence, Theorem~\ref{thm:tensorization} applied to
 $T_{1-\eta} f$ implies that
\[
 \E J_\rho(T_{1-\eta} f(X), T_{1-\eta} f(Y)) \le J_\rho(\E T_{1-\eta} f, \E T_{1-\eta} f) + \Delta_n(f)
 \le J_\rho(\E f, \E f) + C \epsilon^{-C(\rho)} \tau^{c\eta}.
\]
Since $J_\rho(x, y) \ge xy$, it follows that
\[
 \Stab_{\rho(1-\eta)^2}(f) = \E T_{1-\eta} f(X) T_{1-\eta} f(Y) \le J_\rho(\E f, \E f) + C \epsilon^{-C(\rho)} \tau^{c\eta}.
\]
This inequality holds for any $0 < \rho < 1$; hence we can replace $\rho(1-\eta)^2$ by $\rho$ to obtain
\begin{equation}\label{eq:maj-is-stablest-1}
 \Stab_\rho(f) = \E T_{1-\eta} f(X) T_{1-\eta} f(Y) \le J_{\rho(1-\eta)^{-2}}(\E f, \E f) + C \epsilon^{-C(\rho)} \tau^{c\eta}
\end{equation}
for any $\rho \le (1-\eta)^2$.

Now,~\eqref{eq:maj-is-stablest-1} holds for any $f: \{-1, 1\}^n \to [\epsilon, 1-\epsilon]$.
For a function $f$ taking values in $[-1, 1]$, let $f^\epsilon$ be $f$ truncated to $[\epsilon, 1-\epsilon]$.
Since $|\E f^\epsilon - \E f| \le \epsilon$ and (by the proof of Claim~\ref{clm:negative-semidefinite})
$\pdiff{J_\rho(x,y)}{x} \le 1$ for any $\rho$,
\[
 J_{\rho(1-\eta)^{-2}}(\E f^\epsilon, \E f^\epsilon) \le 
 J_{\rho(1-\eta)^{-2}}(\E f, \E f) + 2 \epsilon.
\]
On the other hand, $|f - f^\epsilon| \le \epsilon$ and so
\[
 \Stab_\rho(f) = \E f(X) f(Y) \ge \Stab_\rho(f^\epsilon) - 2 \epsilon.
\]
Thus,~\eqref{eq:maj-is-stablest-1} applied to $f^\epsilon$ implies that
for any $\rho \le (1-\eta)^2$ and any $\epsilon > 0$,
\[
 \Stab_\rho(f) \le J_{\rho(1-\eta)^{-2}}(\E f, \E f) + 2 \epsilon + C \epsilon^{-C(\rho)} \tau^{c\eta}.
\]
If we set $\epsilon = \tau^{c\eta/(2C(\rho))}$ then
\[
 \Stab_\rho(f) \le J_{\rho(1-\eta)^{-2}}(\E f, \E f) + C \tau^{c(\rho)\eta}.
\]

Finally,
some calculus on $J_\rho$ (see Claim~\ref{clm:diff-J-rho}) shows that $|\pdiff{J_\rho(x,y)}{\rho}| \le (\sqrt{1-\rho^2})^{-3/2}$
for any $x, y$; hence
\[
 \Stab_\rho(f) \le J_{\rho(1-\eta)^{-2}}(\E f, \E f)
 + \frac{(1-\eta)^{-2} - 1}{(1-\rho^2)^{3/2}} + C \tau^{c(\rho)\eta}
 \le J_\rho(\E f, \E f) + C(\rho) (\eta + \tau^{c(\rho) \eta}).
\]
Choosing $\eta = C(\rho) \frac{\log \log(1/\tau)}{\log(1/\tau)}$ completes the proof
as long as $\rho \le (1-\eta)^2$. However, we can trivially make the theorem true for
$(1-\eta)^2 \le \rho$ by choosing $C(\rho)$ and $c(\rho)$ appropriately.
\end{proof}

%
%


\section{SoS proof of Majority is Stablest}\label{sec:majstable}

The principal theorem of this section is the  SoS version of ``Majority is Stablest" theorem of \cite{MOO:10}. Before we state the theorem, we will need a few definitions. We will consider the indeterminates $f(x)$ (for $x \in \{-1,1\}^n$). The  constraints on these indeterminates is given by $$A_ p = \{ 0 \le f(x) \le  1 : \textrm{ for all } x \in \{-1,1\}^n\}$$. 
As is the case with the usual setting, its helpful to define the fourier coefficients of $f$. 
$$
\textrm{For } S \subseteq [n] \quad  \widehat{f}(S) = \mathop{\mathbf{E}}_{x \in \{-1,1\}^n} f(x) \cdot \chi_S(x)  \quad \quad \textrm{ and hence } f(x)  = \sum_{S \subset[n]} \widehat{f}(S) \chi_S(x)
$$
Note that $\widehat{f}(S)$ are nothing but linear forms in terms of the original indeterminates. It is also helpful to recall the notion of influences and low-degree influences in this context.
$$
\Inf_i(f)  = \sum_{S \ni i} \widehat{f}^2(S)  \quad \quad \Inf_i^{\le d}(f)  = \sum_{S \ni i : |S| \le d} \widehat{f}^2(S)
$$
With this, we state the main theorem of this section. 
\begin{theorem}\label{thm:SoS-maj}
For any $\kappa>0$ and $\rho \in (-1,0)$, $\exists d_0= d_0(\kappa,\rho)$, $d_1=d_1(\kappa,\rho)$ and $c = c(\kappa,\rho)$ such that 
$$
A_p \vdash_{d_0} \mathop{\mathbf{E}}_{\substack{x \in \{-1,1\}^n \\ y \sim_{\rho} x}}  [f(x) \cdot f(y) + (1-f(x)) \cdot (1- f(y))]   \ge 1- \frac{1}{\pi} \arccos \rho - \kappa - c \cdot (\sum_{i=1}^{n} (\Inf_{i}^{\le d_1}f)^2)
$$
\end{theorem}
This is easily seen to be equal to the statement of the Majority is Stablest theorem of \cite{MOO:10}. Before, we delve further into the SoS proofs, we feel its good to  familiarize ourselves with the fourier machinery in the SoS world. The upshot of the discussion ahead is going to be that the basic fourier identities and operations hold without any changes in the SoS world. In particular, it is easy to verify that Parseval's identity holds i.e. for $\{f(x)\}$ and $\{\widehat{f}(S)\}$ defined as above 
$
\mathbf{E}[f^2(x)] = \sum_{S \subseteq [n]} \widehat{f}^2(S)
$.

Similarly, we can define the noise operator $T_{\rho}$ here as follows : Given the sequence of indeterminates $\{f(x) \}_{x \in \{-1,1\}^n}$, we define the sequence of indeterminates $\{g(x)\}_{x \in \{-1,1\}^n}$ as 
$
g(x) = \mathbf{E}_{y \sim_{\rho} x} [f(x)] 
$
and for every $x$, use $T_{\rho} f(x)$ to refer to $g(x)$. It is also easy to check that if we define $\widehat{g}(S) = \mathbf{E}_{x} [g(x) \cdot \chi_S(x)]$, then 
$\widehat{g}(S) = \rho^{|S|} \widehat{f}(S)$. 

\subsection{Smoothening the function}
For our purposes, it is necessary to do a certain smoothening of the function $f$. In particular, we start by considering a new function $f_2$ i.e. we create a new sequence of indeterminates defined by $f_2(x) = (1-\epsilon) f(x) + \epsilon/2$ for some $\epsilon>0$. The value of $\epsilon$ shall be fixed later.   We observe that 
\begin{equation}\label{eq:rand2}
A_p \vdash_1 \cup_{x \in \{-1,1\}^n} \{ \epsilon \le f_1(x) \le 1-\epsilon \} \end{equation} \begin{equation*}
\widehat{f_1}(S) = (1-\epsilon) \hat{f}(S) + (\epsilon/2) \cdot \mathbf{1}_{S = \Phi} 
\end{equation*}
We also make the following claim (the proof is deferred to Appendix~\ref{app:missing-SoS-proof}).  
\begin{claim}\label{clm:smooth1}
$$
A_p \vdash_2 f(x) f(y) -2 \epsilon \le f_1(x) f_1(y) \le f(x) f(y) + 2\epsilon 
$$
\end{claim}
The next stage of smoothening is done by defining $g = T_{1-\eta} f_2$  for some $\eta>0$. Again, the value of $\eta$ will be fixed later. 
\begin{equation}\label{eq:rand1}
 \cup_{x \in \{-1,1\}^n} \{ \epsilon \le f_1(x) \le 1-\epsilon \}  \vdash_1 \cup_{x \in \{-1,1\}^n} \{ \epsilon \le g(x) \le 1-\epsilon \} \end{equation} \begin{equation*}
 \widehat{g}(S) = (1-\epsilon) (1-\eta)^{|S|} \hat{f}(S) + (\epsilon/2) \cdot \mathbf{1}_{S = \Phi} \end{equation*}
Also, observe that $\mathbf{E}_{x, y \sim_{\rho} x} [f_2(x) \cdot f_2(y)] = \mathbf{E}_{x, y \sim_{\rho'} x} [g(x) \cdot g(y)]$ where $\rho' = \rho/(1-\eta)^2$. Of course, this imposes the condition $|\rho| < |1-\eta|^2$. So, we have to choose $\eta$ to be small enough. Now, define the constraint set $A'_p =\cup_{x \in \{-1,1\}^n} \{  \epsilon \le g(x) \le 1- \epsilon \}$. So, we summarize the discussion of this subsection in the following two claims. 
\begin{claim}\label{clm:subsec}
For any $q$ and $d \in \mathbb{N}$, if $A'_p \vdash_d q \ge 0$, then $A_p \vdash_d q \ge 0$. 
\end{claim}
The proof of the above is obtained by combining (\ref{eq:rand2}) and (\ref{eq:rand1}) with the third bullet of Fact~\ref{fac:SoS-basic0}. The second claim is
\begin{claim}\label{clm:subsec1}
$$
A_p \vdash_2 \mathbf{E}_{x, y \sim_{\rho} x} [f(x) \cdot f(y)] \ge \mathbf{E}_{x, y \sim_{\rho'} x} [g(x) \cdot g(y)] - 2\epsilon
$$
\end{claim}
Thus, the above two claims mean that from now on, we will work with $A'_p$ and aim to prove a lower bound on  $\mathbf{E}_{x, y \sim_{\rho'} x} [g(x) \cdot g(y)]$. At this stage, let $\tilde{J}_{\rho'}$ be the approximation obtained from Claim~\ref{clm:bernstein} with parameter $\epsilon>0$ and $\delta= \epsilon$. For the sake of brevity, we indicate this by $\tilde{J}$ itself. The following claim allows us to compare the terms $x\cdot y$ and $\tilde{J}(x,y)$. 
\begin{claim}\label{lem:comparison}
For any  $\epsilon>0$,  such that $\rho' \in (-1,0)$ and $\tilde{J}$ is as described above, there is a $d_{\alpha} = d_{\alpha}(\epsilon,  \rho')$ such that, 
$$
\{ \epsilon \le x\le 1- \epsilon, \epsilon \le y \le 1 - \epsilon \} \vdash_{d_{\alpha}}  x \cdot y \ge \tilde{J}(x,y) - 2\epsilon
$$
\end{claim}
\begin{proof}
Note that for $(x,y) \in (0,1)^2$, $J_{0}(x,y) = xy$ and hence by Slepian's lemma, we get that if $\rho' <  0$, then $xy \ge J_{\rho'}(x,y)$. Now, by definition, we have  that for $(x,y) \in [\epsilon, 1- \epsilon]^2$
, $
xy \ge \tilde{J} (x,y) - \epsilon
$. In other words, if we define the polynomial $p(x,y) = xy -\tilde{J}(x,y) +2 \epsilon$, then we know that  for $(x,y) \in [\epsilon, 1- \epsilon]^2$, $p(x,y) \ge \epsilon$. We can thus apply Corollary~\ref{corr:Putinar} to get that there is an integer $d_\alpha=d_\alpha(\epsilon,  \rho')$  such that for $(x,y) \in [\epsilon, 1- \epsilon]^2$, 
for $\rho' \in (0,1)$,
 $$
\{ \epsilon \le x\le 1- \epsilon, \epsilon \le y \le 1 - \epsilon \} \vdash_{d_1}  p \ge 0
$$
Expanding $p$, finishes the proof. 
\end{proof}

\subsection{Taylor's theorem in the SoS world} 

Following the proof of Majority is Stablest, we now need to prove a Taylor's theorem in the SoS hierarchy. 
The following lemma is the SoS analogue of Claim~\ref{clm:base-case}. 
\begin{lemma}\label{lem:taylor}
Define a sequence of indeterminates $\{h_0(1), h_0(-1), h_1(1), h_1(-1) \}$. Let $A$ be a set of constraints defined as $A=\mathop{\cup}_{i,j \in \{0,1\}}\{  \epsilon \le h_i(j) \le 1-\epsilon \}$.
For any $\epsilon>0, \rho' \in (-1,0)$,  $\exists c_\gamma= c_\gamma(\epsilon , \rho')$ and $\exists d_\gamma= d_\gamma(\epsilon , \rho')$  such that  
$$
A \vdash_{d_\gamma} \mathop{\mathbb{E}}_{\substack{x \in_R \{-1,1\} \\ y \sim_{\rho'} x}} [\widetilde{J}(h_0(x), h_1(y))] \ge \widetilde{J}(\widehat{h_0}(0), \widehat{h_1}(0)) -\epsilon \cdot ( \widehat{h_0}^2(1) + \widehat{h_1}^2(1) )  - c_{\gamma} \cdot(\widehat{h_0}^4(1) + \widehat{h_1}^4(1))
$$
where $\widehat{h_i}(j) = \frac{h_i(0) + (-1)^j \cdot h_i(1)}{2}$ for $i, j \in \{0,1\}$. 
\end{lemma}
\begin{proof}
We start by noting that since $\widetilde{J}$ is a symmetric polynomial, hence we can write 
$$
\widetilde{J}(x,y) = \sum_{m, n : m + n \le K} \mu_{\{m,n\}} x^m y^n
$$
Here, we assume that $K$ is the degree of $\widetilde{J}$ and $c$ is the maximum absolute value of any coefficient. .   We next make the following claim. 
\begin{claim}\label{clm:simplify}
$$ \mathop{\mathbb{E}}_{\substack{x \in_R \{-1,1\} \\ y \sim_{\rho'} x}} [\widetilde{J}(h_0(x), h_1(y))]  =  \sum_{m,n : m+ n \textrm{ is even}} \nu_{m,n} \cdot \widehat{h_0}^m(1)  \cdot\widehat{h_1}^n(1) \cdot \left( \frac{1+\rho'}{2} + (-1)^m \cdot \frac{1-\rho'}{2} \right) 
$$
where
$$
\nu_{m,n} = \sum_{m_1 \ge m ; n_1 \ge n} \mu_{m_1, n_1} \cdot \widehat{h_0}^{m_1-m}(0) \cdot \widehat{h_1}^{n_1-n}(0) \cdot \binom{m_1}{m} \binom{n_1}{n}
$$
\end{claim}
The proof  is deferred to Appendix~\ref{app:missing-SoS-proof}. Next, we note that $\nu_{0,0} = \widetilde{J}(\widehat{h_0}(0), \widehat{h_1}(0))$. 
Thus, we get that 
\begin{equation}\label{eq:difference-terms}
 \mathop{\mathbb{E}}_{\substack{x \in_R \{-1,1\} \\  y \sim_{\rho'} x}} [\widetilde{J}(h_0(x), h_1(y))]  - \tilde{J}(\widehat{h_0}(0), \widehat{h_1}(0)) =  \sum_{\substack{m,n : m+ n \textrm{ is even} \\ K \ge m + n \ge 2}} \nu_{m,n} \cdot \widehat{h_0}^m(1) \cdot \widehat{h_1}^n(1) \cdot \left( \frac{1+\rho'}{2} + (-1)^m \cdot \frac{1-\rho'}{2} \right) 
\end{equation}
We first make the following claim which bounds the terms when $m  + n \ge 4$. 
\begin{claim}\label{clm:bound1}
$$
A \vdash_{4K+3} Y \ge  \sum_{\substack{m,n : m+ n \textrm{ is even } \\ \textrm{ and }m+n \ge 4}} \nu_{m,n}  \cdot \widehat{h_0}^m(1) \cdot \widehat{h_1}^n(1) \cdot \left( \frac{1+\rho'}{2} + (-1)^m \cdot \frac{1-\rho'}{2} \right)  \ge -Y
$$
where $Y = 2cK^4 2^{2K} (\widehat{h_0}^4(1) + \widehat{h_1}^4(1))$. 
\end{claim}
Again, we defer the proof of Claim~\ref{clm:bound1} to Appendix~\ref{app:missing-SoS-proof}.
Thus, we are only left with the task of controlling the terms when  $m+n=2$. Note that 
$$
\sum_{m+n=2} \nu_{m,n} \cdot \widehat{h_0}^m(1) \cdot  \widehat{h_1}^n(1) \cdot \left( \frac{1+\rho'}{2} + (-1)^m \cdot \frac{1-\rho'}{2} \right) = \nu_{2,0} \cdot \widehat{h_0}^2(1) + \nu_{0,2} \cdot \widehat{h_1}^2(1) + \rho' \nu_{1,1} \cdot \widehat{h_0}(1) \widehat{h_1}(1)
$$ For the sake of brevity, call the above quantity $\Lambda$. 
Next, we observe that at $x=\widehat{h_0}(0), y = \widehat{h_1}(0)$
$$
\frac{\partial^2 \widetilde{J}(x,y) }{\partial x^2}   = 2\nu_{2,0} \ \  \ \frac{\partial^2 \widetilde{J}(x,y) }{\partial y^2}   = 2\nu_{0,2}  \ \ \ \frac{\partial^2 \widetilde{J}(x,y) }{\partial x \partial y}   = \nu_{1,1}
$$
To see this, note that 
$$
\frac{\partial^2 \widetilde{J}(x,y) }{\partial x^2}_{\substack{x=\widehat{h_0}(0)  , y = \widehat{h_1}(0)}} = \frac{\partial^2 \widetilde{J}(\widehat{h_0}(0) +x',\widehat{h_1}(0)+y') }{\partial x'^2}_{x'=0  , y'=0} 
$$
However, the quantity on the right side is simply twice the coefficient of $x'^2$ in the polynomial $ \widetilde{J}(\widehat{h_0}(0) +x',\widehat{h_1}(0)+y') $ which is exactly $2\nu_{2,0}$. The other equalities follow similarly. 
Thus, we get that 
$$
\Lambda= \frac{1}{2} \left( \frac{\partial^2 \tilde{J}(x,y) }{\partial x^2}\widehat{h_0}^2(1) + \frac{\partial^2 \tilde{J}(x,y) }{\partial y^2} \widehat{h_1}^2(1) +2 \rho'\frac{\partial^2 \tilde{J}(x,y) }{\partial x \partial y}\widehat{h_0}(1) \widehat{h_1}(1) \right) 
$$
In the above, all the derivatives are evaluated at $x=\widehat{h_0}(0), y = \widehat{h_1}(0)$. 
We now make the following claim which gives a lower bound on $\Lambda$. 
\begin{claim}\label{clm:lower-bound-Lambda}
For every $\epsilon > 0$ and $\rho \in (-1,0)$,   there exists $ d_\gamma' = d_\gamma'(\epsilon , \rho')$ such that  
$$
A \vdash_{ d_\gamma' } \Lambda \ge - \epsilon \cdot ( \widehat{h_0}^2(1) + \widehat{h_1}^2(1) )  
$$
\end{claim}
We defer this proof to Appendix~\ref{app:missing-SoS-proof}.  Now, set $d_{\gamma} = \max \{d'_{\gamma}, 4K+3\}$ and $c_{\gamma} =2cK^4 2^{2K}$.  Combining Claim~\ref{clm:bound1} and Claim~\ref{clm:lower-bound-Lambda} with (\ref{eq:difference-terms}), we get Lemma~\ref{lem:taylor}.
\end{proof}

\subsection{Tensorization: } We now do a ``tensorization" of the inequality in Lemma~\ref{lem:taylor}. Let $\{\phi(x) \}_{x \in \{-1,1\}^n}$ be a set of indeterminates. We recall that for $y \in \{-1,1\}^i$, we define the set $\{\phi_y(z) \}_{z \in \{-1,1\}^{n-i}}$ of indeterminates  as follows : $\phi_y(z) = \phi(z \cdot y)$.  
As before, we can define the fourier coefficients $\widehat{\phi_y}(S)$ for $S \subseteq [n-i]$ and it is easy to see that they are homogenous linear forms in the indeterminates $\widehat{\phi}(S)$ (for $S \subseteq [n]$). We now state a few basic properties for the indeterminates $ g_y(z)$ and $\widehat{g_y}(S)$.
\begin{align}
\label{eq:prelim-obs1}
A'_p &\vdash_1 \mathop{\cup}_{i=0}^{n-1}  \mathop{\cup}_{y \in \{-1,1\}^i}  \mathop{\cup}_{z \in \{-1,1\}^{n-i}} \{ \epsilon \le g_y(z) \le 1- \epsilon \}
\\
\label{eq:prelim-obs2}
A'_p &\vdash_1 \mathop{\cup}_{i=0}^{n-1}  \mathop{\cup}_{y \in \{-1,1\}^i}  \mathop{\cup}_{S \subseteq [n-i]} \{ -1 \le \widehat{g_y}(S) \le 1 \}
\end{align}
\begin{align}
\label{eq:prelim-obs3}
\vdash_2\mathop{\mathbf{E}}_{y \in \{-1,1\}^i} [\widehat{g_x}^2(n-i)] = \sum_{\substack{S \subseteq \{n-i , \ldots  , n \}  \\ n -i \in S }} \widehat{g}^2(S)
\end{align}


\begin{lemma}\label{lem:roll-out-induction}
For the parameters $c_\gamma = c_\gamma(\epsilon, \rho')$ and $d_{\gamma} = d_{\gamma} (\epsilon, \rho')$ from Lemma~\ref{lem:taylor}, 
\begin{eqnarray*}A'_p \vdash_{d_\gamma} \mathop{\mathbf{E}}_{\substack{x \in \{-1,1\}^n \\  y \sim_{\rho'} x}}  [\widetilde{J}(g(x), g(y))] &\ge& \widetilde{J}(\mathbf{E}[g(x)], \mathbf{E}[g(y)]) - \epsilon \left(\sum_{i=0}^{n-1} \mathbf{E}_{z \in \{-1,1\}^{i}} [\widehat{g}_z^2(n-i)]\right) \\ &-& c_{\gamma} \left(\sum_{i=0}^{n-1} \mathbf{E}_{z \in \{-1,1\}^{i}} [\widehat{g}_z^4(n-i)] \right)  \end{eqnarray*}
\end{lemma}
The proof of this claim is a very simple induction. For the sake of completeness, we give the proof in Appendix~\ref{app:missing-SoS-proof}. 
We now simplify the error terms. Towards this, note that (\ref{eq:prelim-obs3}) implies that 
$$
\vdash_2 \epsilon \left(\sum_{i=0}^{n-1} \mathbf{E}_{z \in \{-1,1\}^{i}}[\widehat{g}_z^2(n-i)] \right)= \sum_{S \not = \phi} \widehat{g}^2(S) \le \sum_{S } \widehat{g}^2(S)  = \mathop{\mathbf{E}}_{x \in \{-1,1\}^n} [g^2(x)]
$$
Further, $A'_p \vdash_3 \mathop{\mathbf{E}}_{x \in \{-1,1\}^n} [g^2(x)] \le 1$ (using Fact~\ref{fac:SoS0}). 
Thus, we get that \begin{equation}\label{eq:penultimate}
A'_p  \vdash_{d_\gamma} \mathop{\mathbf{E}}_{\substack{x \in \{-1,1\}^n \\ y \sim_{\rho'} x}}  [\widetilde{J}(g(x), g(y))] \ge \widetilde{J}(\mathbf{E}[g(x)], \mathbf{E}[g(y)]) - \epsilon -c_{\gamma} \left(\sum_{i=0}^{n-1} \mathbf{E}_{x \in \{-1,1\}^{i}} [\widehat{g}_x^4(n-i)]\right) 
\end{equation}

\subsection{Bounding the error terms:} 
Thus, all we are left to bound is the ``degree-4" term. We briefly describe why one has to be careful to get a (meaningful) upper bound here. The reason is that the obvious strategy to do this is to break $g$ into high degree and low-degree parts based on the noise parameter (call them $h$ and $\ell$). Now, this very naively gives an error term of the form $\mathbf{E}_x \widehat{h}_x^4(n-i)$ and $\mathbf{E}_x \widehat{\ell}_x^4(n-i)$. The latter can be easily bound using hypercontractivity. However,  there does not seem to be obvious way to bound the former. This is in spite of the fact that $\mathbf{E}_x \widehat{h}_x^2(n-i)$ is small. We now show how to get around this problem. 

We define $d_{\eta} = (1/\eta) \cdot \log (1/\eta)$. Now, define the sequence of indeterminates $\{h(x) \}_{x \in \{-1,1\}^n}$ and $\{\ell(x) \}_{x \in \{-1,1\}^n}$ as follows : 
$$ h (x)= \sum_{|S|>d_\eta} \widehat{g}(S) \chi_S(x) \quad \quad \ell (x) = \sum_{|S| \le d_\eta} \widehat{g}(S) \chi_S(x)$$
By the way it is defined, it is clear that $\vdash_1 h(x) + \ell(x) = g(x)$. 
Now, we can analyze the term $\mathbf{E}_{x \in \{-1,1\}^{i}} [\widehat{g}_x^4(n-i)]$ as
\begin{eqnarray*}
\vdash_4 \sum_{i=0}^{n-1} \mathbf{E}_{x \in \{-1,1\}^{i}} [\widehat{g}_x^4(n-i)] &=&  \sum_{i=0}^{n-1} (\mathbf{E}_{x \in \{-1,1\}^{i}} [\widehat{g}_x^3(n-i) (\widehat{h}_x(n-i) + \widehat{\ell}_x(n-i))] )\\ 
&=&  \sum_{i=0}^{n-1} ( \mathbf{E}_{x \in \{-1,1\}^{i}} [\widehat{g}_x^3(n-i) \widehat{h}_x(n-i)] +  \mathbf{E}_{x \in \{-1,1\}^{i}} [\widehat{g}_x^2(n-i) \widehat{\ell}_x^2(n-i)] \\ 
&+& \mathbf{E}_{x \in \{-1,1\}^{i}} [\widehat{g}_x^2(n-i) \widehat{\ell}_x(n-i) \widehat{h}_x(n-i)] )
\end{eqnarray*}
We begin by stating the following useful fact : 
\begin{fact}\label{fac:temp}
$A_p \vdash_3 \sum_{i=0}^{n-1} \mathbf{E}_{x \in \{-1,1\}^i} \widehat{h}_x^2(n-i)  \le \eta$
\end{fact}
\begin{proof}
$$
 \vdash_2 \sum_{i=0}^{n-1} \mathbf{E}_{x \in \{-1,1\}^i} \widehat{h}_x^2(n-i) = \sum_{S} \widehat{h}^2(S)  = \sum_{|S| > d_\eta} (1-\eta)^{d_{\eta}}  \widehat{f}^2(S) \le  \eta \cdot (\sum_{|S| > d_\eta}  \widehat{f}^2(S) ) \le  \eta \cdot (\sum_{|S|}  \widehat{f}^2(S) )
$$
$$
A_p \vdash_3 \left(\sum_{|S|}  \widehat{f}^2(S) \right) = \mathbf{E}_{x \in \{-1,1\}^n} [f^2(x)] \le 1
$$
Combining the two facts, finishes the proof. 
\end{proof}

We now  make the following claims. 
\begin{claim}\label{clm:final-1}
$A_p \vdash_6 \sum_{i=0}^{n-1}  \mathop{\mathbf{E}}_{x \in \{-1,1\}^{i}} [\widehat{g}_x^3(n-i) \widehat{h}_x(n-i)]  \le \sqrt{\eta}$. 
\end{claim}
\begin{claim}\label{clm:final-2}
$
A_p \vdash_5 \sum_{i=0}^{n-1} \mathop{\mathbf{E}}_{x \in \{-1,1\}^{i}} [\widehat{g}_x^2(n-i) \widehat{\ell}_x^2(n-i)] \le \sqrt{\eta} + \frac{9^{d_\eta}}{\sqrt{\eta}} ( \sum_{i=1}^n (\Inf_i^{\le d_\eta}(f))^2)
$
\end{claim}
\begin{claim}\label{clm:final-3}
$
A_p \vdash_8 \sum_{i=0}^{n-1} \mathbf{E}_{x \in \{-1,1\}^{i}} [\widehat{g}_x^2(n-i) \widehat{h}_x(n-i) \widehat{\ell}_x(n-i)] \le 2\sqrt{\eta} + \frac{9^{d_\eta}}{\sqrt{\eta}} ( \sum_{i=1}^n (\Inf_i^{\le d_\eta}(f))^2)
$
\end{claim}
The proofs are deferred to the appendix. Combining (\ref{eq:penultimate}) with Claim~\ref{clm:final-1}, Claim~\ref{clm:final-2}, Claim~\ref{clm:final-3} along (\ref{eq:penultimate}) and Claim~\ref{clm:subsec}, we get that for $c_{\gamma}$ and $d_{\gamma}$ described in Lemma~\ref{lem:taylor},
\begin{eqnarray*}
A_p \vdash_{d_{\gamma}} \mathop{\mathbf{E}}_{\substack{x \in \{-1,1\}^n \\  y \sim_{\rho'} x}}  [\widetilde{J}(g(x), g(y))] &\ge& \widetilde{J}(\mathbf{E}[g(x)], \mathbf{E}[g(y)]) - \epsilon   - 4 \cdot c_{\gamma} \sqrt{\eta}\\ &-& \frac{2 \cdot 9^{d_{\eta}} \cdot c_{\gamma}}{\sqrt{\eta}} \left( \sum_{i=1}^n (\Inf_i^{\le d_\eta}(f))^2 \right)
\end{eqnarray*}

Using Claim~\ref{lem:comparison}, we have that $A_p \vdash_{d_\alpha} g(x) \cdot g(y) \ge   \widetilde{J}(g(x), g(y)) - 2\epsilon$. Similarly, combining this with Claim~\ref{clm:subsec1}, we can get that
\begin{equation}\label{eq:aux1}
A_p \vdash_{d_\alpha} \mathbf{E}_{x, y \sim_{\rho} x} [f(x) \cdot f(y)] \ge   \mathbf{E}_{x, y \sim_{\rho'} x} [\widetilde{J}(g(x), g(y))]  - 4 \epsilon
\end{equation}
Thus, now applying (\ref{eq:aux1}), we get
\begin{equation}\label{eq:borell}
A_p \vdash_{\max\{d_{\gamma},d_{\alpha}\}} \mathop{\mathbf{E}}_{\substack{x \in \{-1,1\}^n \\  y \sim_{\rho'} x}}  [f(x) \cdot f(y)] \ge \widetilde{J}(\mathbf{E}[g(x)], \mathbf{E}[g(y)]) - 5 \epsilon   - 4 \cdot c_{\gamma} \sqrt{\eta}- \frac{2 \cdot 9^{d_{\eta}} \cdot c_{\gamma}}{\sqrt{\eta}} \left( \sum_{i=1}^n (\Inf_i^{\le d_\eta}(f))^2 \right)
\end{equation}
Now, define a new sequence of indeterminates $\{f_c(x) \}_{x \in \{-1,1\}^n}$ where $f_c(x) =1-f(x)$.  
Next, define $f_{2c} (x) = (1-\epsilon) f_c(x)   + \epsilon/2$. Next, we define $g_c(x) = \mathbf{E}_{y \sim_{1-\eta} x} [f_{2c}(x)]$. 
We now make the following observations : 
\begin{itemize}
\item $\forall x \in \{-1,1\}^n$, $A_p \vdash_1 \forall x \in \{-1,1\}^n$, $\epsilon \le g_c(x) \le (1-\epsilon)$. 
\item $\mathbf{E}_{x} [g(x)] + \mathbf{E}_x [g_c(x)]=1$. 
\item For all $i \in [n]$, $\Inf_i^{\le d_\eta} f =  \Inf_i^{\le d_\eta} f_c$. 
\end{itemize}
Thus, using the above, analogous to (\ref{eq:borell}), we have the following : 
\begin{equation}\label{eq:borellc}
A_p \vdash_{\max\{d_{\gamma},d_{\alpha}\}}  \mathop{\mathbf{E}}_{\substack{x \in \{-1,1\}^n\\ y \sim_{\rho} x}}  [f_c(x) \cdot f_c(y)] \ge \widetilde{J}(\mathbf{E}[g_c(x)], \mathbf{E}[g_c(y)]) - 5 \epsilon   - 4 \cdot c_{\gamma} \sqrt{\eta}- \frac{2 \cdot 9^{d_{\eta}} \cdot c_{\gamma}}{\sqrt{\eta}} \left( \sum_{i=1}^n (\Inf_i^{\le d_\eta}(f_c))^2 \right)
\end{equation}
Now,  define $\xi $ as $\xi =  5 \epsilon   - 4 \cdot c_{\gamma} \sqrt{\eta}- \frac{2 \cdot 9^{d_{\eta}} \cdot c_{\gamma}}{\sqrt{\eta}} \left( \sum_{i=1}^n (\Inf_i^{\le d_\eta}(f))^2 \right)$. Summing up (\ref{eq:borell}) and (\ref{eq:borellc}), we get
\begin{align}
A_p \vdash_{\max\{d_{\gamma},d_{\alpha}\}}  &\mathop{\mathbf{E}}_{\substack{x \in \{-1,1\}^n \\ y \sim_{\rho} x}}  [f(x) \cdot f(y) + (1-f(x)) \cdot (1- f(y))] \nonumber  \\ \label{eq:rho1} &\ge  \widetilde{J}(\mathbf{E}[g(x)], \mathbf{E}[g(y)]) + \widetilde{J}(\mathbf{E}[1-g(x)], \mathbf{E}[1-g(y)])  - 2\xi\end{align}
Next, we recall the following  fact :
\begin{fact}\label{fac:maj-stable}
For any $a \in (0,1)$ and $\rho \in (-1,0)$, 
$$
 J_{\rho} (a,a) +  J_{\rho} (1-a,1-a)  \ge 2 J_{\rho}(1/2, 1/2) =  1- \frac{ \arccos \rho}{\pi}
$$
\end{fact}
Combining this fact with Claim~\ref{clm:bernstein}, we have that for every $x \in [\epsilon, 1-\epsilon]$, 
$$
\widetilde{J}(x,x) + \widetilde{J}(1-x, 1-x) \ge 1- \frac{ \arccos \rho'}{\pi} - 2\epsilon
$$
By using Corollary~\ref{corr:Putinar}, we have that there exists $d_{\delta} = d_{\delta}(\epsilon, \rho')$ such that
\begin{equation}\label{eq:rho}
A_p \vdash_{d_{\delta}} \widetilde{J}(\mathbf{E}[g(x)], \mathbf{E}[g(y)]) + \widetilde{J}(\mathbf{E}[1-g(x)], \mathbf{E}[1-g(y)])  \ge  1- \frac{ \arccos \rho'}{\pi} - 4\epsilon
\end{equation}
Combining (\ref{eq:rho1}) and (\ref{eq:rho}), we get 
\begin{align}\label{eq:final-manipulate}
A_p \vdash_{\max\{d_{\gamma},d_{\alpha} , d_{\delta}\}}\mathop{\mathbf{E}}_{\substack{x \in \{-1,1\}^n \\ y \sim_{\rho} x}}  [f(x) \cdot f(y) + (1-f(x)) \cdot (1- f(y))]  &\ge 1- \frac{ \arccos \rho'}{\pi} - 14 \epsilon   - 8 \cdot c_{\gamma} \sqrt{\eta} \nonumber \\ &- \frac{4 \cdot 9^{d_{\eta}} \cdot c_{\gamma}}{\sqrt{\eta}} \left( \sum_{i=1}^n (\Inf_i^{\le d_\eta}(f))^2 \right)
\end{align}
From here, getting to Theorem~\ref{thm:SoS-maj} is pretty easy. We proceed as follows : \begin{itemize} \item  For the given $\rho$ and $\kappa$, first we choose $\epsilon = \kappa/100$. This  implies that $14 \epsilon \le \kappa/4$. \item Next, observe that $c_{\gamma}(\rho', \epsilon)$ is a uniformly continuous function of $\rho'$ and $\epsilon$. Now, recall that $\rho' = \rho/(1-\eta)$. Hence,  there exists $\eta_0  = \eta_0 (\rho , \epsilon , \kappa)$ such that for all $\eta \le \eta_0$, $\sqrt{\eta} \cdot c_{\gamma}(\rho', \epsilon) \le \kappa/32$. 
\item Again, observe that for any $\rho \in (-1,0)$ and $\kappa>0$,  there exists $\eta_1 = \eta_1( \rho , \kappa)$ such that for all $\eta \le \eta_1$, $(\arccos \rho')/\pi \le (\arccos \rho')/\pi + \kappa/4$. 
\end{itemize}
Now, choose $\eta = \min \{\eta_0, \eta_1\}$. With $\eta$ and $\epsilon$ having been fixed in terms of $\kappa$ and $\rho$, we set $d_0(\kappa,\rho) = \max\{d_{\gamma},d_{\alpha} , d_{\delta}\}$, $c(\kappa, \rho) = \frac{4 \cdot 9^{d_{\eta}} \cdot c_{\gamma}}{\sqrt{\eta}} $  and $d_1(\kappa, \rho)  = d_{\eta}$ and hence get 

$$
A_p \vdash_{d_0(\kappa,\rho)}\mathop{\mathbf{E}}_{\substack{x \in \{-1,1\}^n \\ y \sim_{\rho} x}}  [f(x) \cdot f(y) + (1-f(x)) \cdot (1- f(y))]  \ge 1- \frac{ \arccos \rho}{\pi} - \kappa- c(\kappa, \rho) \cdot \left( \sum_{i=1}^n (\Inf_i^{\le d_1(\kappa, \rho)}(f))^2 \right) $$
This finishes the proof of Theorem~\ref{thm:SoS-maj}.

\section{Refuting the Khot-Vishnoi instances of \textsf{MAX-CUT}} 
In this section, we will prove the following theorem : 
\begin{theorem}\label{thm:Maj-ref}
Let $\rho \in (-1,0)$ and $G_{\rho} = (V_{\rho}, E_{\rho})$ be the Max-Cut instance constructed in \emph{\cite{KV05}} for the noise parameter $\rho$. Let $\{x_v\}_{v \in V}$ be a sequence of indeterminates 
and $A = \cup_{v \in V} \{0 \le x_v \le 1\}$. Then, for any $\delta>0$, there exists $d_1 = d_1(\delta, \rho)$ such that
\begin{align*}
 A \cup \left\{ \mathop{\mathbf{E}}_{(u,v) \in E} x_u \cdot (1-x_v) + x_v \cdot (1-x_u)  \ge \frac{1}{\pi} \arccos \rho + \delta \right\}
\vdash_{d_1} -1 \ge 0
\end{align*}
\end{theorem}

For this section, it is helpful  to begin by recalling the following theorem  of O'Donnell and Zhou~\cite{OZ:12}. 
\begin{theorem}\label{thm:OD}
\emph{\cite{OZ:12}} Let $\{f(x)\}_{x \in \{-1,1\}^n}$ be a sequence of indeterminates and let $A = \cup_{x \in \{-1,1\}^n} \{ 0 \le f(x) \le 1 \}$. Then, for any $\delta >0$ and $\rho \in (-1,0)$, 
$$
A \vdash_{O(1/\delta^2)}   \mathrm{Stab}_{\rho}(f) \ge  K(\rho)  -\delta - 2^{O(1/\delta^2)} \cdot \left( \sum_{i=1}^n \widehat{f}^4(i)\right)
$$
where $K(\rho) = \frac{1}{2} + \frac{\rho}{\pi} + \left(\frac{1}{2}  - \frac1\pi \right) \cdot \rho^3$
\end{theorem}
The main application of Theorem~\ref{thm:OD} in \cite{OZ:12} is the following : Khot and Vishnoi~\cite{KV05} construct instances of \textsf{MAX-CUT} (parameterized by noise parameter $\rho$) whose optimum is bounded by $(\arccos \rho)/\pi+ o(1)$ (and yet the basic SDP relaxation for \textsf{MAX-CUT} has value $(1-\rho)/2$.) O'Donnell and Zhou essentially use Theorem~\ref{thm:OD} as a black-box to give a constant degree SoS proof that these instances have optimum bounded by $1 - K(\rho)+ o(1)$. This improves significantly on the bound provided by the basic SDP. 

In this section, we will show how we can use the stronger version of Theorem~\ref{thm:OD}, namely Theorem~\ref{thm:SoS-maj} to do even better. In particular, we will use this theorem to give a constant degree SoS proof that these \textsf{MAX-CUT} instances have optimum bounded by $(\arccos \rho)/\pi+ o(1)$ (which is of course tight).  We will not give all the details of this proof as our proof will follow the (by now, standard) reduction from \cite{KKMO07} and its SoS variant from \cite{OZ:12}. 


 We begin by recalling the description of instances of \textsf{UNIQUE-GAMES} (UG). A UG instance is specified by a set of vertices $V$ and an alphabet $[k]$. Along with this, there is a probability distribution $\mathcal{E}$ on tuples of the form $(u,v ,\pi_{(u,v)})$ with $\pi_{(u,v)} : [k] \rightarrow [k]$ being a permutation. Further, the weighted graph defined by $\mathcal{E}$ is regular. Also, let $\mathcal{E}_u$ denote the marginal distribution on $(v, \pi)$ when the first vertex is conditioned to be $u$. 
  The objective is to get a mapping $L : V \rightarrow [k]$ so as to maximize the following quantity : 
$$
\Pr_{(u,v, \pi_{(u,v)}) \in \mathcal{E}} [\mathcal{L}(v) = \pi_{(u,v)}(\mathcal{L}(u))]
$$
We next consider the SoS formulation for the  UG instance described above. It is slightly different from the ``obvious" formulation and follows the formulation in \cite{OZ:12}. In particular, we define variables $x_{v,i}$ for every $v \in V$ and $i \in [k]$.  Now, consider the set of constraints defined by $$A_p = \bigcup_{v \in V, i \in [k]}  \left\{x_{v,i} \ge 0  \right\} \bigcup_{v \in V}   \left\{\sum_{i \in [k]} x_{v,i} \le 1 \right\}$$ 
It is easy to show that if the optimum solution to the Unique Games instance is bounded by $\beta$, then $$\mathop{\mathbf{E}}_{u \in V} [ \sum_{i=1}^k ( \mathop{\mathbf{E}}_{(v, \pi_{u,v}) \in \mathcal{E}_u} x_{v, \pi_{u,v}(i)})^2] \le 4\beta$$ for any set of indeterminates $\{ x_{v,i} \}$  which obeys the constraint set $A_p$.  We now make the following definition :
\begin{definition}\label{def:refutation}
Given a UG instance $(V, \mathcal{E})$ with alphabet size k,  there is a degree-$d$ SOS refutation for optimum  $\beta$ if
$$
A_p \bigcup \left\{ \mathop{\mathbf{E}}_{u \in V} \left[ \sum_{i=1}^k \left( \mathop{\mathbf{E}}_{(v, \pi_{u,v}) \in \mathcal{E}_u} x_{v, \pi_{u,v}(i)}\right)^2\right] \ge \beta \right\} \vdash_d -1 \ge 0
$$
\end{definition}
Before we go ahead, we recall that for any $\eta \in (0,1)$ and $N \in \mathbb{N}$ (which is a power of $2$),  \cite{KV05} construct UG instances over $2^N/N$ vertices, alphabet size $n$ such that optimal value of the instance is bounded by $N^{-\eta}$. \footnote{Of course, the interesting part is that \cite{KV05} shows that the standard SDP relaxation on this instance has value $1-\eta$} Modifying the result from \cite{BBHKSZ:12}, O'Donnell and Zhou~\cite{OZ:12} show the following : 
\begin{theorem}\label{thm:UG-refutation}
Let $\eta \in (0,1)$ and $N$ be a power of $2$ and let $(V, \mathcal{E})$ be the corresponding instances of UG constructed in \cite{KV05}. Then, there is a degree-$4$ SoS refutation for optimum $\beta = N^{-\Omega(\eta)}$. 
\end{theorem}

We next describe the reduction from \cite{KKMO07} of UG to \textsf{MAX-CUT}.  The reduction is parameterized by a ``correlation" value $\rho \in (-1,0)$. Given the instance of UG described above,  the set of vertices in the corresponding \textsf{MAX-CUT} instance is given by $V'=V \times \{-1,1\}^k$. Further, the probability distribution $\mathcal{E}_{\rho,k}$ over the edges is given by the following sampling procedure : 
\begin{itemize}
\item Choose $u \sim V$ uniformly at random. 
\item Choose $(u,v_1 , \pi_{(u,v_1)}) $  and $(u,v_2 , \pi_{(u,v_2)}) $ independently from the distribution $E_u$ which is defined as the marginal of $E$ conditioned on the first vertex being $u$. 
\item Choose $x \in \{-1,1\}^k$ and $y \sim_{\rho} x$. 
\item  Output vertices $((v_1, \pi_{(u,v_1)}(x)),(v_2, \pi_{(u,v_2)}(y)))$
\end{itemize}
Now, for a function $g : \{-1,1\}^k \rightarrow [0,1]$, let us define $\mathrm{Stab}_{\rho} (g)$ as follows $$ \mathrm{Stab}_{\rho}(g) = \mathop{\mathbf{E}}_{x \in \{-1,1\}^k, y \sim_{\rho} x} [g(x) \cdot g(y) + (1- g(x)) \cdot (1-g(y))]$$
We have the following simple claim : 
\begin{claim}\label{clm:max-cut-value}
\emph{\cite{KKMO07}} Let $G' = (V', E_{\rho,k})$ be an instance of \textsf{MAX-CUT} described above. Consider a partition of the graph $G'$ (into two sets) specified by a collection of functions $\{f_v : \{-1,1\}^k \rightarrow \{0,1\} \}$. Then, the value of cut defined by this partition is $1 - \mathop{\mathbf{E}}_{u \in V} [\mathrm{Stab}_{\rho}(g_u)]$ where
$$
g_u : \{-1,1\}^k \rightarrow [0,1] \ \textrm{ is defined as } \ g_u(x) = \mathop{\mathbf{E}}_{(v,\pi) \in  \mathcal{E}_u}[f_v(\pi (x))]
$$
\end{claim}
Consider the SoS relaxation of the \textsf{MAX-CUT} instance defined by $V'$ and $\mathcal{E}_{\rho, k}$. In particular,  we have an indeterminate $f_v(z)$ for every $v \in V$ and $z \in \{-1,1\}^k$. The constraint set $A_m$ is given by $A_m = \cup_{v \in V} \cup_{z \in \{-1,1\}^k} \{0 \le f_v(z) \le 1\}$. Then, O'Donnell and Zhou~\cite{OZ:12} show that if $(V,\mathcal{E})$ is a UG instance such that there is a degree $4$ refutation for the optimum $\beta$, then
\begin{equation}\label{eq:odz}
A_m \cup \{1 - \mathop{\mathbf{E}}_{u \in V} [\mathrm{Stab}_{\rho}(g_u)] \ge K(\rho)  -\delta - 2^{O(1/\delta^2)} \beta\} \vdash_{O(1/\delta^2) + 4} -1 \ge 0
\end{equation}
This of course means that the if the \cite{KKMO07} reduction is applied on the instances from Theorem~\ref{thm:UG-refutation}, 
$$
A_m \cup \{1 - \mathop{\mathbf{E}}_{u \in V} [\mathrm{Stab}_{\rho}(g_u)] \ge K(\rho)  -\delta - 2^{O(1/\delta^2)} \cdot N^{-\Omega(\eta)} \} \vdash_{O(1/\delta^2) + 4} -1 \ge 0
$$

Exactly following the same steps as \cite{OZ:12}, but using Theorem~\ref{thm:SoS-maj} instead of Theorem~\ref{thm:OD}, we show that for any $\delta>0$, 
\begin{equation}
A_m \cup \{1 - \mathop{\mathbf{E}}_{u \in V} [\mathrm{Stab}_{\rho}(g_u)] \ge (\arccos \rho)/\pi  -\delta - d_2( \delta, \rho) c(\delta, \rho) \cdot \beta\} \vdash_{d_1(\delta, \rho) + 4} -1 \ge 0
\end{equation}
We do not repeat the steps here and leave it to the reader to fill the details. Using $\beta = N^{-\Omega(\eta)}$, we get Theorem~\ref{thm:Maj-ref}. 

\section*{Acknowledgements}
We thank Ryan O'Donnell and Yuan Zhou for sharing the manuscript~\cite{OZ:12}. 
\bibliography{allrefs,all,mossel}

\appendix

\section{Facts regarding $J_\rho$}

Here we collect various facts about the function
\[
 J_\rho(x, y) = \Pr[X \le \Phi^{-1}(x), Y \le \Phi^{-1}(y)],
\]
where $(X, Y) \sim \mathcal{N}(0, (\begin{smallmatrix} 1 & \rho \\ \rho & 1\end{smallmatrix}))$.
These calculations all follow from elementary calculus.

\negsemi*
\begin{proof}

Towards proving this, note that we can define $Y = \rho \cdot X + \sqrt{1-\rho^2} \cdot Z$ where $Z \sim \mathcal{N}(0,1)$ is an independent normal.  Also, let us define $\Phi^{-1}(x) = s$ and $\Phi^{-1}(y) =t$. For $s, t \in \mathbb{R}$, define $K_{\rho}(s,t)$ as
$$
K_{\rho}(s,t) = \Pr_{X,Y} [X \le s, Y \le t] = \Pr_{X,Z} [ X \le s, Z \le (t - \rho \cdot X)/\sqrt{1-\rho^2}]
$$
Note that for the aforementioned relations between $x$, $y$, $s$ and $t$, $K_{\rho}(s,t) = J_{\rho}(x,y)$.  Note that 
\begin{equation}\label{eq:K}
K_{\rho}(s,t) = \int_{s'=-\infty}^{s} \phi(s') \int_{t'=-\infty}^{(t - \rho \cdot s')/\sqrt{1-\rho^2}} \phi(t') ds' dt'
\end{equation}
This implies that 
$$
\frac{\partial K_{\rho}(s,t)}{\partial s} =  \phi(s) \int_{t'=-\infty}^{(t - \rho \cdot s)/\sqrt{1-\rho^2}} \phi(t')  dt'
$$
By chain rule, we get that 
$$
\frac{\partial J_{\rho}(x,y)}{\partial x} = \frac{\partial K_{\rho}(s,t)}{\partial s} \cdot  \frac{\partial s}{\partial x}
$$
By elementary calculus, it follows that 
$$
\frac{d \Phi^{-1}(x)}{dx} = \frac{1}{\phi (\Phi^{-1}(x))} \quad \Rightarrow \quad  \frac{\partial s}{\partial x} = \frac{1}{\phi(\Phi^{-1}(x))} =\frac{1}{\phi(s)}
$$
Thus, 
$$
\frac{\partial J_{\rho}(x,y)}{\partial x} = \int_{t'=-\infty}^{(t - \rho \cdot s)/\sqrt{1-\rho^2}} \phi(t')  dt'
$$
Thus, we next get that 
$$
\frac{\partial^2 J_{\rho}(x,y)}{\partial x^2}  =\frac{\partial^2 J_{\rho}(x,y)}{\partial x \partial s}  \cdot  \frac{\partial s}{\partial x} = \phi \left(\frac{t-\rho \cdot s}{\sqrt{1-\rho^2}} \right) \cdot \frac{-\rho}{\sqrt{1-\rho^2}} \cdot \frac{1}{\phi(s)} =\phi \left(\frac{\Phi^{-1}(y)-\rho \cdot \Phi^{-1}(x)}{\sqrt{1-\rho^2}} \right) \cdot \frac{-\rho}{\sqrt{1-\rho^2}} \cdot \frac{1}{\phi(s)}  $$
$$
\frac{\partial^2 J_{\rho}(x,y)}{\partial x \partial y} =\frac{\partial^2 J_{\rho}(x,y)}{\partial x \partial t} \cdot  \frac{\partial t}{\partial y} = \phi \left(\frac{\Phi^{-1}(y)-\rho \cdot \Phi^{-1}(x)}{\sqrt{1-\rho^2}} \right) \cdot \frac{1}{\sqrt{1-\rho^2}} \cdot \frac{1}{\phi(t)}
$$
Because we know that $(X,Y) \sim (Y,X)$, by symmetry, we can conclude that 
$$
\frac{\partial^2 J_{\rho}(x,y)}{\partial y^2}   =\phi \left(\frac{\Phi^{-1}(x)-\rho \cdot \Phi^{-1}(y)}{\sqrt{1-\rho^2}} \right) \cdot \frac{-\rho}{\sqrt{1-\rho^2}} \cdot \frac{1}{\phi(t)}  $$
and likewise, 
$$
\frac{\partial^2 J_{\rho}(x,y)}{\partial y \partial x} = \phi \left(\frac{\Phi^{-1}(x)-\rho \cdot \Phi^{-1}(y)}{\sqrt{1-\rho^2}} \right) \cdot \frac{1}{\sqrt{1-\rho^2}} \cdot \frac{1}{\phi(s)}$$
It is obvious now that 
$$
\pdiffII {J_\rho(x,y)}xx \cdot \pdiffII {J_\rho(x,y)}yy
- \rho^2 \left(\pdiffII {J_\rho(x,y)}xy\right)^2 = 0.
$$
Now, suppose that $|\sigma| \le |\rho|$. Then
$$
\det(M_{\rho \sigma}(x, y))
=
\pdiffII {J_\rho(x,y)}xx \cdot \pdiffII {J_\rho(x,y)}yy
- \sigma^2 \left(\pdiffII {J_\rho(x,y)}xy\right)^2 \ge 0.
$$
If $\rho \ge 0$ then the diagonal of $M_{\rho \sigma}(x,y)$ is non-positive, and it follows that
$M_{\rho \sigma}(x, y)$ is negative semidefinite. If $\rho \le 0$ then
the diagonal is non-negative and so $M_{\rho \sigma}(x,y)$ is positive semidefinite.
\end{proof}

\thirddiff*
\begin{proof}
As before, we set $\Phi^{-1}(x) =s $ and $\Phi^{-1}(y) =t$. From the proof of Claim~\ref{clm:negative-semidefinite}, we see that 
$$
\frac{\partial^2 J_{\rho}(x,y)}{\partial x^2}  =\phi \left(\frac{\Phi^{-1}(y)-\rho \cdot \Phi^{-1}(x)}{\sqrt{1-\rho^2}} \right) \cdot \frac{-\rho}{\sqrt{1-\rho^2}} \cdot \frac{1}{\phi(s)}  
$$
To compute the  third derivatives of $J$, recall that
  $\frac{\partial s}{\partial x} = \frac{1}{\phi(s)}$ and $\frac{\partial t}{\partial y}= \frac{1}{\phi(t)}$, we have
 \begin{eqnarray*}
   \frac{\partial^3 J_{\rho}(x,y)}{\partial x^3} &=& \frac{\rho}{(1-\rho^2)^{3/2}}
   \frac{\rho t + (2\rho^2 - 1) s}{\phi(s)}
   \exp\Big(-\frac{t^2 - 2\rho st + (2\rho^2 - 1) s^2}{2(1-\rho^2)}\Big) \notag \\
   &=& \frac{\sqrt{2\pi} \rho}{(1-\rho^2)^{3/2}}
   (\rho t + (2\rho^2 - 1) s)
   \exp\Big(-\frac{t^2 - 2\rho st + (3\rho^2 - 2) s^2}{2(1-\rho^2)}\Big).
   \label{eq:Jaaa}
 \end{eqnarray*}
 Now, $\Phi^{-1}(x) \sim \sqrt{2 \log x}$ as $x \to 0$; hence there is a constant $C$
  such that $\Phi^{-1}(x) \le C \sqrt{\log x}$ for all $x \le \frac{1}{2}$.
  Hence, $\exp(s^2) \le x^{-C}$ for all $x \le \frac{1}{2}$; by symmetry,
  $\exp(s^2) \le (x(1-x))^{-C}$ for all $x \in (0, 1)$.
  Therefore
  \begin{eqnarray*}
   \exp\Big(-\frac{t^2 - 2\rho st + (3\rho^2 - 2)s^2}{2(1-\rho^2)}\Big)
   &=&
    e^{-\frac{t^2}{2(1-\rho^2)}} e^{\frac{\rho st}{1 - \rho^2}} e^{\frac{(2 - 3 \rho^2) s^2}{2(1-\rho^2}} \notag\\
   &\le&
    e^{-\frac{t^2}{2(1-\rho^2)}} e^{\frac{\rho (s^2 + t^2)}{2(1 - \rho^2)}} e^{\frac{(2 - 3 \rho^2) s^2}{2(1-\rho^2}} \notag\\
   &\le&
    \big(x(1-x)y(1-y)\big)^{-\frac{\rho}{2(1-\rho^2)}}
    \big(x(1-x)\big)^{-\frac{2 - 3 \rho^2}{2(1-\rho^2)}} \notag\\
   &\le& \big(x(1-x)y(1-y)\big)^{-C(\rho)}.
\label{eq:Jaaa-bound}
  \end{eqnarray*}
  Applying this to~\eqref{eq:Jaaa}, we see that
  $$\left|\frac{\partial^3 J_{\rho}(x,y)}{\partial x^3}\right| \le C(\rho) \big(x(1-x)y(1-y)\big)^{-C(\rho)}.$$
  The other third derivatives are similar:
  \[
  \frac{\partial^3 J_{\rho}(x,y)}{\partial x^2 \partial y}
   = \frac{\sqrt{2\pi} \rho}{(1-\rho^2)^{3/2}}
   (t - 2 \rho s)
   \exp\Big(-\frac{(2\rho^2 - 1) t^2 - 2\rho st + (2\rho^2 - 1) s^2}{2(1-\rho^2)}\Big).
  \]
  By the same steps that led to~\eqref{eq:Jaaa-bound}, we get
  $$\left|\frac{\partial^3 J_{\rho}(x,y)}{\partial x^2 \partial y}\right| \le C(\rho) \big(x(1-x)y(1-y)\big)^{-C(\rho)}$$ (for a slightly different $C(\rho)$).
  The bounds on $\partial^3 J/\partial y^2 \partial x$ and $\partial^3 J/\partial x^3$ then follow because $J$ is symmetric in $x$ and $y$.
 \end{proof}

 \begin{claim}\label{clm:diff-J-rho}
 For any $x, y \in (0, 1)$,
  \[
   \left|\pdiff{J_\rho(x, y)}{\rho}\right| \le (1-\rho^2)^{-3/2}.
  \]
 \end{claim}
 \begin{proof}
  We begin from~\eqref{eq:K}, but this time we differentiate with respect to $\rho$:
  \[
   \pdiff{K_\rho(s,t)}{\rho} = -\frac{1}{(1 - \rho^2)^{3/2}}
   \int_{s'=-\infty}^s \phi(s') \phi\left(\frac{t - \rho s'}{\sqrt{1-\rho^2}}\right) ds'.
  \]
 Since $\phi \le 1$ and $\int_{s'} \phi(s') ds' = 1$, it follows that
  \[
   \left|\pdiff{K_\rho(s, t)}{\rho}\right| \le (1-\rho^2)^{-3/2}.
  \]
 Since $\pdiff{J_\rho(s,t)}{\rho} = \pdiff{K_\rho(\Phi^{-1}(x), \Phi^{-1}(y))}{\rho}$,
 the proof is complete.
 \end{proof}

\section{Approximation by polynomials}

\begin{claim}\label{clm:bernstein}
For any $\rho \in (-1, 1)$ and any $\delta > 0$, there is a polynomial
$\tilde J$ such that for all $0 \le i + j \le 2$,
\[
 \sup_{x,y \in [\epsilon,1-\epsilon]} \left|
 \frac{\partial^{i+j} J_\rho(x,y)}{\partial x^i \partial y^j}
 - \frac{\partial^{i+j} \tilde J_\rho(x,y)}{\partial x^i \partial y^j}
 \right| \le \delta.
\]
Moreover, if $\rho \in [-1+\epsilon, 1-\epsilon]$, then the
degree of $\tilde J$ and the maximal coefficient in $\tilde J$
can be bounded by constants depending only on $\epsilon$ and $\delta$.
\end{claim}

The proof of Claim~\ref{clm:bernstein} follows from standard
results on Bernstein polynomials. In particular, we make
use of the following theorem which may be found, for example, in~\cite{Lorentz:86}.

\begin{theorem}\label{thm:bernstein}
Suppose $f: [0, 1] \to \R$ has $m$ continuous derivatives
which are all bounded in absolute value by $M$.
For any $n \in \N$, let $B_{n}f$ be the polynomial
\[
 (B_n f)(x) = \sum_{k=1}^n f(k/n) \binom{n}{k} x^k (1-x)^{n-k}.
\]
Then for any $0 \le i \le m$,
\[
 \sup_{x \in [0, 1]}
 \left|\frac{d^i f(x)}{dx^i} - \frac{d^i (B_n f)(x)}{dx^i}\right| \le C \sqrt{M/n}.
\]
\end{theorem}

Seeing as the first three derivatives of $J_\rho$ are bounded on $[\epsilon, 1-\epsilon]$,
Claim~\ref{clm:bernstein} is essentially just a 2-variable version of Theorem~\ref{thm:bernstein}.
Although such a result is almost certainly known (and for more than 2 variables), we were not
to find a reference in the literature, and so we include the proof here.

\begin{proof}[Proof of Claim~\ref{clm:bernstein}]
Suppose that $f: [0, 1]^2 \to \R$ has all partial derivatives up to third order bounded by $M$.
Define
\begin{align*}
 g_n(x, y)
 &= (B_n f(\cdot, y))(x) = \sum_{k=1}^n \binom{n}{k} f(k/n, y) x^k(1-x)^{n-k} \\
 h_n(x, y)
 &= (B_n g_n(x, \cdot))(y) = \sum_{k=1}^n \sum_{\ell=1}^n \binom{n}{k} \binom{n}{\ell} f(k/n, \ell/n)
  x^k(1-x)^{n-k} y^\ell (1-y)^{n-\ell}.
\end{align*}
Fix $0 \le i + j \le 2$ and note that
\begin{align}
\label{eq:g_n-diff}
 \frac{\partial^j g_n(\cdot, y)}{\partial y^j} &= B_n \frac{\partial^j f(\cdot, y)}{\partial y^j} \\
\label{eq:h_n-diff}
 \frac{\partial^i h_n(x, \cdot)}{\partial x^i} &= B_n \frac{\partial^i g_n(x, \cdot)}{\partial x^i}.
\end{align}
Now,
fix $y \in [0, 1]$ and apply Theorem~\ref{thm:bernstein} to~\eqref{eq:g_n-diff}:
for any $x \in [0, 1]$,
\[
 \left|\frac{\partial^{i+j} g_n(x, y)}{\partial x^i \partial y^j}
 - \frac{\partial^{i+j} f(x, y)}{\partial x^i \partial y^j}\right| \le C \sqrt{M/n}.
\]
On the other hand, fixing $x$ and applying Theorem~\ref{thm:bernstein} to~\eqref{eq:h_n-diff} yields
\[
 \left|\frac{\partial^{i+j} h_n(x, y)}{\partial x^i \partial y^j}
 - \frac{\partial^{i+j} g_n(x, y)}{\partial x^i \partial y^j}\right| \le C \sqrt{M/n}.
\]
Putting these together,
\[
 \left|\frac{\partial^{i+j} h_n(x, y)}{\partial x^i \partial y^j}
 - \frac{\partial^{i+j} f(x, y)}{\partial x^i \partial y^j}\right| \le 2C \sqrt{M/n}.
\]
Since $h_n$ is a polynomial, taking $n$ sufficiently large implies that there is a polynomial $\tilde f$ such that
$\tilde f$, and its partial derivatives of order at most 2, uniformly approximate
the corresponding derivatives of $f$.
Although we stated this for functions on $[0, 1]^2$, a change of coordinates shows that it holds equally well for
functions on $[\delta, 1-\delta]^2$ with three bounded derivatives. Since $J_\rho$ is such a function,
the first part of the claim follows.

For the second part of the claim, note that all of the error bounds hold uniformly in $\rho \in [-1+\epsilon, 1-\epsilon]$ since the third
derivatives of $J_\rho$ are uniformly bounded over $\rho \in [-1+\epsilon, 1-\epsilon]$. Moreover, since $\max_{x,y} |J_\rho(x, y)| \le 1$, the coefficients in $h_n$ can be bounded in terms of $n$, which is in
turn bounded in terms of $\epsilon$ and $\delta$.
\end{proof}

\section{Useful facts in SoS hierarchy}\label{app:SoS}
\begin{fact}\label{fac:deg-increase}
If $A \vdash_d p\ge  0$ and $A \vdash q \ge 0$, then $A \vdash_d p+q \ge 0$. 
\end{fact}
\begin{fact}\label{fac:SoS0} If $A = \{-1\le y \le 1\}$,
\begin{itemize}
\item If $k$ is an even integer, $A \vdash_{k+1} 0 \le y^k \le 1$
\item If $k$ is an odd integer,  $A \vdash_{k}  -1 \le y^{k} \le 1$
\end{itemize}
\end{fact}
\begin{proof}
Note that for $k=0,1$, the conclusion is trivially true. For $k=2$, note that trivially, $y^2 \ge 0$. So, 
we begin by observing that (from \cite{OZ:12}) 
$$
1- y^2 = \frac{1}{2} (1+y)^2 (1-y) + \frac{1}{2} (1+y) (1-y)^2 
$$
and hence $0 \le y \le 1 \vdash_{3} y^2 \le 1$. This finishes the case for $k=2$. 
For the remaining cases, we use induction. We first consider the case when $k >2$ is even. Then, trivially, we have $y^k \ge 0$. Also, observe that $
1- y^{k} = y^2 ( 1- y^{k-2}) + (1-y^2) 
$. Hence, by induction hypothesis, we have $A \vdash_{k+1}  y^k \le 1$. 

Next, consider the case when $k>2$ is odd. Again, as $1- y^{k} = y^2 ( 1- y^{k-2}) + (1-y^2)$, hence by induction hypothesis, we get $A \vdash_{k}   y^{k} \le 1$. Also, note that 
$1+y^k = y^2 ( 1+ y^{k-2}) + (1-y^2)$. Hence, again, by induction hypothesis, we get $A \vdash_{k}   -1 \le y^{k} $
\end{proof}
\begin{fact}\label{fac:SoS-10}
$-1\le y \le 1 \vdash_5 y^4 \le y^2$
\end{fact}
\begin{proof}
$$
y^2- y^4 = \frac{1}{2} y^2 (1+y)^2 (1-y) + \frac{1}{2} (1+y) (1-y)^2 y^2
$$
This finishes the proof.
\end{proof}
\begin{fact}\label{fac:SoS-1}
Let $a \le y \le b \vdash_d p(y) \ge 0$. Then, for $\lambda_1, \ldots, \lambda_k \in \mathbb{R}$ such that $\forall i \in [k]$, $\lambda_i \ge 0$ and $\sum_{i=1}^k \lambda_i=1$, we have that $\{ \cup_{i=1}^k 0 \le Z_i \le 1 \} \vdash_d p(\sum_{i=1}^k \lambda_i Z_i ) \ge 0$. 
\end{fact}
\begin{proof}
In the SoS proof of $p(Y) \ge 0$, whenever the term $(b-Y)$ appears, we simply substitute it by $\sum_{i=1}^k \lambda_i (b- Z_i)$. Likewise,  whenever the term $(Y-a)$ appears, we substitute it by $\sum_{i=1}^k \lambda_i (a-Z_i)$. It is easy to see that this substitution shows that $\{ \cup_{i=1}^k 0 \le Z_i \le 1 \} \vdash_d p(\sum_{i=1}^k \lambda_i Z_i ) \ge 0$. 
\end{proof}
\begin{fact}\label{fac:SoS-2}
For $k\ge 3$, $-1 \le y \le 1 \vdash_{2k+1} 0 \le y^{2k} \le y^4$
\end{fact}
\begin{proof}
$
1- y^2 = \frac{1}{2} (1+y)^2 (1-y) + \frac{1}{2} (1+y) (1-y)^2 
$
and hence $-1 \le y \le 1 \vdash_{3} y^2 \le 1$. As a consequence, for any $j\ge 1$. we can get that $-1 \le y \le 1 \vdash_{2j+1} y^{2j} \le y^{2j-2}$. Summing all the inequalities as $j$ variables from $j=3$ to $j=k$, we get the stated inequalities.
\end{proof}
\begin{fact}\label{fac:SoS1}
For integers $m,n \ge 2$, $\{-1 \le y \le 1 , -1 \le z \le 1 \}\vdash_{1+ \max\{2m,2n\}} -(y^4 + z^4) \le y^m z^n \le (y^4 + z^4)$. 
\end{fact}
\begin{proof}
$\vdash_{\max\{2m,2n\}} y^{2m} + z^{2n}  \ge y^m z^n$. Also, using Fact~\ref{fac:SoS-2},  $-1 \le y \le 1 \vdash_{2m+1} y^4 \ge y^{2m}$. Similarly, we have $-1 \le z \le 1 \vdash_{2m+1} z^4 \ge z^{2m}$. Combining these, we have $\{-1 \le y \le 1 , -1 \le z \le 1 \}\vdash_{1+ \max\{2m,2n\}}  y^m z^n \le (y^4 + z^4)$. Replacing $y$ by $-y$ and $z$ by $-z$, we can similarly get, $\{-1 \le y \le 1 , -1 \le z \le 1 \}\vdash_{1+ \max\{2m,2n\}}  -y^m z^n \le (y^4 + z^4)$. This completes the proof.
\end{proof}
\begin{fact}\label{fac:SoS2}
For any odd integer $n \ge 3$, $\{-1 \le y \le 1 , -1 \le z \le 1 \}\vdash_{n+2} -(y^4 + z^4) \le yz^n \le (y^4 + z^4)$. 
\end{fact}
\begin{proof}
We use $A$ to denote $ \{-1 \le y \le 1 , -1 \le z \le 1 \}$. 
We first use Fact~3.10 from \cite{OZ:12} which states that $A \vdash_4 yz^3 \le y^4 + z^4$.   We can replace $y$ by $-y$ to get the other inequality. This already gives the proof for $n=3$. For $n>3$, we have that $A \vdash_{n+1} yz^n \le y^4 z^{n-3} + z^{n+1}$. Now, using Fact~\ref{fac:SoS0} (Item 1), we get $A \vdash_{n+2} z^{n+1} \le z^4$. And similarly, we get $A \vdash_{n-2} z^{n-3} \le 1$ and hence $A \vdash_{n+2} z^{n-3} y^4 \le y^4$. Combining these, we get that $A \vdash_{n+2} yz^n \le y^4 + z^4$. Replacing $y$ by $-y$, we get the other side. 
\end{proof}
\begin{fact}\label{fac:SoS3}
Let $A \vdash_{d_1} 0 \le x \le 1$ and $A \vdash_{d_2} -z \le y \le z$ where $z \in \mathbb{M}_{\le d}[X]$. Then, $A \vdash_{d_1 + \max\{d_2,d_3\}}  -z \le xy \le z$. 
\end{fact}
\begin{proof}
Note that $z-xy = z(1-x) + x(z-y)$. Now, $A \vdash_{d_1 +d_2} x(z-y) \ge 0$ and $A \vdash_{d_3 + d_1} z(1-x) \ge 0$. Combining these, we get that  $A \vdash_{d_1 + \max\{d_2,d_3\}}   xy \le z$. Flipping $y$ to $-y$, we get the other inequality. 
\end{proof}
\begin{fact}\label{fac:hyper}\emph{\cite{BBHKSZ:12}}
Let $n, d \in \mathbb{N}$ and $d \le n$. For every $S \subseteq [n]$ such that $|S| \le d$, we have an indeterminate $\widehat{\ell}(S)$. For $x \in \{-1,1\}^n$, define $\ell (x) = \sum_{S \subseteq [n] : |S| \le d} \widehat{\ell}(S) \chi_S(x)$. Then, 
$$
\vdash_4 \mathop{\mathbf{E}}_{x \in \{-1,1\}^n} [\ell^4(x)] \le 9^d \left( \mathop{\mathbf{E}}_{x \in \{-1,1\}^n} [\ell^2(x)] \right)^2 
$$
\end{fact}
\subsection*{Consequences of Putinar's Positivstellensatz}

The following is a consequence of Theorem~\ref{thm:Putinar}.
\begin{corollary}\label{corr:Putinar}
Let $X = (x_1, x_2)$ and $A = \{x_1 \ge \epsilon, x_2 \ge \epsilon, x_1 \le 1-\epsilon, x_2 \le 1- \epsilon \}$. Then, for any $p(X)$ such that $p(X) \ge \epsilon$ on $\mathbb{V}(A)$, there exists an integer $d = d(p)$ such that $A \vdash_d p \ge \epsilon/2$. \end{corollary} 
\begin{proof}
We can define the polynomials $p_1 = x_1 - \epsilon$, $p_2 = 1- \epsilon -x_1$,  $p_3 = x_2 - \epsilon$ and $p_4 = 1-\epsilon - x_2$. Now, note that $q(x,y)$ defined as
\begin{eqnarray*}
q(x_1,x_2) &=& (1-\epsilon-x_1)^2 \cdot p_1 + (x_1- \epsilon)^2 \cdot p_2 +  (1-\epsilon-x_2)^2 \cdot p_3 + (x_2- \epsilon)^2 \cdot p_4 \\
&=& \left(1-2\epsilon\right) \left(-\left(x_1-\frac{1}{2}\right)^2 - \left(x_2-\frac{1}{2}\right)^2 +\frac14 - 2\epsilon(1-\epsilon) \right) 
\end{eqnarray*}
Clearly, $q(x_1,x_2) \in \mathbb{M}(S)$ and  that the set $\{(x_1,x_2) : q(x_1, x_2) \ge 0 \}$ is  a compact set. As a consequence, we can apply Theorem~\ref{thm:Putinar} to get that there is an integer $d=d(p)$ such that $S \vdash_d p - \epsilon/2 \ge 0$. This implies $S \vdash p \ge \epsilon/2$. 
\end{proof}
The following is a corollary of Theorem~\ref{thm:matrix-Putinar}. 
\begin{corollary}\label{corr:matrix-Putinar}
Let $X = (x_1, \ldots, x_n)$ and $A = \{p_1(X) \ge 0, \ldots, p_m(X) \ge 0\}$ be satisfying the conditions in Theorem~\ref{thm:Putinar}. Let $\Gamma \in (\mathbb{R}[X])^{p \times p}$ be such that for $x \in \mathbb{V}(A)$, $\Gamma  \succeq \delta I$ for some $\delta > 0$. Let $v \in (\mathbb{R}[X])^{p}$. Then, if $p = v^T \cdot \Gamma \cdot v$, then $p \in \mathbb{M}(A)$. 
\end{corollary}
\begin{proof}
First, by applying Theorem~\ref{thm:matrix-Putinar}, we get $\Gamma =  \Gamma_0(X) + \sum_{i=1}^m \Gamma_i(X) \cdot p_i(X)$. Let us assume that $\Gamma_i = B_i^T \cdot B_i$. Then, 
$$
p= v^T \cdot \Gamma \cdot v = v^T( \Gamma_0(X) + \sum_{i=1}^m \Gamma_i(X) \cdot p_i(X)) v = (B_0 \cdot v)^T \cdot (B_0 \cdot v) +  \sum_{i=1}^m (B_0 \cdot v)^T \cdot (B_0 \cdot v) \cdot p_i(x)$$
This proves the claim.
\end{proof}

\section{Missing proofs from Section~\ref{sec:majstable}}\label{app:missing-SoS-proof}
\begin{proofof}{(of Claim~\ref{clm:smooth1})}
$$
f_1(x) f_1(y) - f(x) f(y) = (\epsilon^2 - 2\epsilon) f(x) f(y)  + \frac{\epsilon^2}{4} + \epsilon (1-\epsilon) (f(x) + f(y))
$$
Hence, $A_p \vdash_2 f_1(x) f_1(y) - f(x) f(y) \le 2\epsilon$. On the other hand, note that $1- f(x) f(y)  = (1-f(x))f(y) + (1-f(y))$ and hence $A_p \vdash_2 f(x) f(y) \le 1$.  This implies that
$ A_p \vdash_2 f(x) f(y) -f_1(x) f_1(y) \le 2 \epsilon $. 
\end{proofof}

\begin{proofof}{(of Claim~\ref{clm:simplify})}
We begin by observing that 
\begin{align*}
\widetilde{J}(\widehat{h_0}(0)+x, \widehat{h_1}(0)+y) &=  \sum_{m,n} \mu_{\{m,n\}} \cdot (\widehat{h_0}(0)+x)^m  \cdot(\widehat{h_1}(0)+y)^n = \sum_{m,n} \nu_{m,n} \cdot x^m \cdot y^n
\end{align*}
As a consequence, we get that 
\begin{align*}
\widetilde{J}(\widehat{h_0}(0)+\widehat{h_0}(1), \widehat{h_1}(0)+\widehat{h_1}(1)) &=  \sum_{m,n} \nu_{m,n} \cdot \widehat{h_0}^m(1) \cdot \widehat{h_1}^n(1)
\\
\widetilde{J}(\widehat{h_0}(0)-\widehat{h_0}(1),\widehat{h_1}(0)-\widehat{h_1}(1)) &= \sum_{m,n} (-1)^{m+n} \nu_{m,n}  \cdot \widehat{h_0}^m(1) \cdot \widehat{h_1}^n(1)
\end{align*}
Adding these equations, we get 
\begin{equation}\label{eq:1}
\widetilde{J}(\widehat{h_0}(0)+\widehat{h_0}(1), \widehat{h_1}(0)+\widehat{h_1}(1)) + \tilde{J}(\widehat{h_0}(0)-\widehat{h_0}(1),\widehat{h_1}(0)-\widehat{h_1}(1)) = 2 \cdot \sum_{\substack{m,n  \\ m+ n \textrm{ is even}}} \nu_{m,n} \cdot \widehat{h_0}^m(1) \cdot \widehat{h_1}^n(1)
\end{equation}
Similarly, we have that 
\begin{align*}
\widetilde{J}(\widehat{h_0}(0)-\widehat{h_0}(1), \widehat{h_1}(0)+\widehat{h_1}(1)) = \sum_{m,n} (-1)^{m} \nu_{m,n} \cdot \widehat{h_0}^m(1) \cdot \widehat{h_1}^n(1)
\\
\widetilde{J}(\widehat{h_0}(0)+\widehat{h_0}(1), \widehat{h_1}(0)-\widehat{h_1}(1)) =  \sum_{m,n} (-1)^{n} \nu_{m,n}  \cdot \widehat{h_0}^m(1) \cdot \widehat{h_1}^n(1)
\end{align*}
Thus, 
\begin{equation}\label{eq:2}
\widetilde{J}(\widehat{h_0}(0)-\widehat{h_0}(1), \widehat{h_1}(0)+\widehat{h_1}(1)) + \tilde{J}(\widehat{h_0}(0)+\widehat{h_0}(1), \widehat{h_1}(0)-\widehat{h_1}(1)) =  2 \cdot \sum_{\substack{m,n  \\ m+ n \textrm{ is even}}} (-1)^m \nu_{m,n} \cdot  \widehat{h_0}^m(1) \cdot \widehat{h_1}^n(1)
\end{equation}
Hence, combining  (\ref{eq:1}) and (\ref{eq:2}), we get that 
$$
\mathop{\mathbf{E}}_{\substack{x \in_R \{-1,1\} \\ y \sim_{\rho'} x}} [\widetilde{J}(h_0(x), h_1(y))] = \sum_{\substack{m,n \\ m+ n \textrm{ is even}}} \nu_{m,n} \cdot \widehat{h_0}^m(1) \cdot \widehat{h_1}^n(1) \cdot \left( \frac{1+\rho'}{2} + (-1)^m \cdot \frac{1-\rho'}{2} \right) 
$$
\end{proofof}

\begin{proofof}{(of Claim~\ref{clm:bound1})}
For $m +n \ge 4$ and $m_1 \ge m$ and $n_1 \ge n$, we define $\Gamma_{m_1,n_1, m , n}$ as follows : 
$$
\Gamma_{m_1,n_1, m , n} = \mu_{m_1, n_1}  \binom{m_1}{m} \binom{n_1}{n} \cdot  \widehat{h_0}^m(1) \cdot \widehat{h_1}^n(1) \cdot \widehat{h_0}^{m_1-m}(0) \cdot  \widehat{h_1}^{n_1-n}(0)\cdot \left( \frac{1+\rho'}{2} + (-1)^m \cdot \frac{1-\rho'}{2} \right) 
 $$
 Next, define the set of constraints $A_m$ as
 $$
 A_m = \{0 \le \widehat{h_0}(0) \le 1 \ , \ 0 \le \widehat{h_1}(0) \le 1 \ , \ -1 \le \widehat{h_0}(1) \le 1 \ , \ -1 \le \widehat{h_1}(1) \le 1  \}
 $$ 
 Now, it is easy to see that $A \vdash_1 A_m$. Hence, by using the third bullet of Fact~\ref{fac:SoS-basic0}, if for any $p$ and $d \in \mathbb{N}$, $A_m \vdash_d p \ge 0$, then $A \vdash_d p \ge 0$. We shall be using this fact consistently throughout this proof.
 Applying Fact~\ref{fac:SoS0}, we get that 
 $$
 A \vdash_{m_1 -m +1} 0 \le \widehat{h_0}^{m_1-m}(0) \le 1 \quad \quad  A \vdash_{n_1 -n +1} 0 \le \widehat{h_1}^{n_1-n}(0) \le 1
 $$
 and hence we have that 
 \begin{equation}\label{eq:hlimit2}
  A \vdash_{m_1+ n_1 -m-n +2} \quad 0 \le \widehat{h_0}^{m_1-m}(0)  \cdot \widehat{h_1}^{n_1-n}(0)\le 1 
 \end{equation}
Now, we consider two possibilities : Either $m, n \ge 2$ or $\max\{m, n \} \ge 3$.  In the first case, using  Fact~\ref{fac:SoS1}, we get that 
\begin{equation}\label{eq:h-limit1}
A \vdash_{1 + \max \{2m, 2n \}}  -(\widehat{h_0}^{4}(1) + \widehat{h_1}^4(1)) \le \widehat{h_0}^m(1) \cdot \widehat{h_1}^n(1) \le \widehat{h_0}^{4}(1) + \widehat{h_1}^4(1) 
\end{equation}
Next, consider the other case i.e.   when $\max\{m,n\} \ge 3$. Without loss of generality, assume $m \ge 3$ and $n=1$. Then, by  Fact~\ref{fac:SoS2}, we get that
$$
A \vdash_{ m+2 }  -(\widehat{h_0}^{4}(1) + \widehat{h_1}^4(1)) \le \widehat{h_0}^m(1) \cdot \widehat{h_1}^n(1) \le \widehat{h_0}^{4}(1) + \widehat{h_1}^4(1) 
$$
Now, combining (\ref{eq:hlimit2}) along with  an application of  Fact~\ref{fac:SoS3}, we get that 
$$
A \vdash_{m+n + m_1+n_1  + 3} -(  \widehat{h_0}^{4}(1) + \widehat{h_1}^4(1) )\le  \widehat{h_0}^m(1) \cdot \widehat{h_1}^n(1) \cdot \widehat{h_0}^{m_1-m}(0) \cdot \widehat{h_1}^{n_1-n}(0) \le \widehat{h_0}^{4}(1) + \widehat{h_1}^4(1) 
$$ 
Now, recalling that $0 \le m_1,  n_1 , m , n  \le K $, $\binom{m_1}{m}, \binom{n_1}{n} \le 2^K$, $|\mu_{m,n}| \le c$ and $|\rho'| \le 1$, we get that
$$
A \vdash_{4K + 3} - c 2^{2K} (  \widehat{h_0}^{4}(1) + \widehat{h_1}^4(1) )\le \Gamma_{m_1,n_1, m , n}  \le c 2^{2K}( \widehat{h_0}^{4}(1) + \widehat{h_1}^4(1))
$$
As, $$\sum_{\substack{m,n : m+ n \textrm{ is even}  \\ \textrm{and }m+n \ge 4}} \nu_{m,n} \widehat{h_0}^m(1) \widehat{h_1}^n(1) \cdot \left( \frac{1+\rho'}{2} + (-1)^m \cdot \frac{1-\rho'}{2} \right) = \sum_{ \substack{m,n : m+ n \textrm{ is even}  \\ m+n \ge 4 \\ K \ge m_1\ge m  \ K \ge n_1 \ge n}}  \Gamma_{m_1,n_1, m , n} $$
As a result, we can conclude that
$$
A \vdash_{4K + 3} - c 2^{2K} K^4 (  \widehat{h_0}^{4}(1) + \widehat{h_1}^4(1) )\le  \sum_{\substack{m,n : m+ n \textrm{ is even}  \\ \textrm{and }m+n \ge 4}} \nu_{m,n} \widehat{h_0}^m(1) \widehat{h_1}^n(1)  \le c 2^{2K} K^4 (  \widehat{h_0}^{4}(1) + \widehat{h_1}^4(1) )
$$
\end{proofof}

\begin{proofof}{(of Claim~\ref{clm:lower-bound-Lambda})}
We begin by defining  the matrix $\widetilde{M}$ as follows : 
$$
\widetilde{M} =   \left(\begin{array}{cc} \frac{\partial^2 \tilde{J}(x,y)}{\partial x^2} & \rho' \frac{\partial^2 \tilde{J}(x,y)}{\partial x \partial y} \\\rho' \frac{\partial^2 \tilde{J}(x,y)}{\partial x \partial y}  & \frac{\partial^2 \tilde{J}(x,y)}{\partial y^2}\end{array}\right)
$$
Put $\beta = 2\epsilon$. Now, let us define $\widetilde{M_1} = \widetilde{M} + \beta I $. Using Claim~\ref{clm:negative-semidefinite} and Claim~\ref{clm:bernstein}, for $\rho' \in (-1,0)$, 
and $(x,y) \in [\epsilon, 1-\epsilon]^2$, we can say that $\widetilde{M_1} \succeq \epsilon \cdot I$. Hence, using Corollary~\ref{corr:matrix-Putinar}, $\exists d_{\gamma}'$ such that we have the following 
$$
A \vdash_{d_{\gamma}'} \widehat{h_0}^2(1) \left( \frac{\partial^2 \tilde{J}(x,y)}{\partial x^2}  + \beta\right) + 2 \rho  \cdot \widehat{h_0}(1)  \cdot\widehat{h_1}(1) \cdot  \frac{\partial^2 \tilde{J}(x,y)}{\partial x \partial y} + \widehat{h_1}^2(1) \left( \frac{\partial^2 \tilde{J}(x,y)}{\partial y^2}  + \beta\right) \ge 0  $$
$$
\equiv A \vdash_{d_{\gamma}'} \widehat{h_0}^2(1) \cdot \frac{\partial^2 \tilde{J}(x,y)}{\partial x^2}   + 2 \rho  \cdot \widehat{h_0}(1) \cdot \widehat{h_1}(1) \cdot  \frac{\partial^2 \tilde{J}(x,y)}{\partial x \partial y} + \widehat{h_1}^2(1)\cdot \frac{\partial^2 \tilde{J}(x,y)}{\partial y^2}  \ge -\beta (\widehat{h_0}^2(1)+\widehat{h_1}^2(1))
$$

Dividing by $2$ on both sides, finishes the proof. Note that the reason the degree $d'_{\gamma}$ depends only on $\epsilon$ and $\rho$ is because from  Corollary~\ref{corr:matrix-Putinar}, the degree $d'_{\gamma}$ depends on $\epsilon$, $\rho$ and  the polynomial $\tilde{J}$ which again in turn depends only on $\epsilon$ and $\rho$. 
\end{proofof}

\begin{proofof}{(of Lemma~\ref{lem:roll-out-induction})}
The proof is by induction.  We introduce the following notation : For any $z \in \{-1,1\}^i$, we use $1 \cdot z \in \{-1,1\}^{i+1}$ to denote the string $z$ with a $1$ prefixed to it. Likewise, we define $-1 \cdot z \in \{-1,1\}^{i+1}$ to denote the string $z$ with a $-1$ prefixed to it.
Next, for any  pairs of strings $z_{-1}, z_1 \in \{-1,1\}^i$ and $j ,k\in \{-1,1\}$, we define the indeterminate, $h_i(j) = \mathbf{E}_{z \in \{-1,1\}^{n-i-1}} [g(z \cdot j \cdot z_i)]$.   Define $A_{z_{-1}, z_{1}} =  \cup_{j ,k\in \{-1,1\}}  \{\epsilon \le h_i(j) \le 1-\epsilon \}$.  It is trivial to see that $A'_p \vdash_1 A_{z_{-1}, z_{1}} $.

 Now, using Lemma~\ref{lem:taylor}, for any two strings $z_1, z_{-1} \in \{-1,1\}^i$, we get that 
 \begin{eqnarray*}
 A_{z_{-1}, z_{1}} &\vdash_{d_\gamma}& \mathop{\mathbf{E}}_{\substack{x \in \{-1,1\} \\ y \sim_{\rho'} x}} \widetilde{J}(\mathbf{E} [g_{x \cdot z_1}], \mathbf{E}[g_{y \cdot z_{-1}}]) 
 \ge \widetilde{J}(\mathbf{E}[g_{z_1}], \mathbf{E}[g_{z_{-1}}])  - \epsilon \left( \widehat{g}_{z_1} ^2(n-i)+\widehat{g}_{z_{-1}} ^2(n-i)\right) \\ &-& c_{\gamma} \left( \widehat{g}_{z_{-1}} ^4(n-i) +\widehat{g}_{z_1} ^4(n-i) \right)
 \end{eqnarray*}
 As a consequence of the third bullet of Fact~\ref{fac:SoS-basic0}, the left hand side of `$\vdash$' can be replaced by $A'_p$.  
 Now, for any given $z_1, z_{-1}  \in \{-1,1\}^{i}$, let $d(z_1, z_{-1})$ be its Hamming distance.  We multiply the inequality by $(\frac{1+\rho'}{4})^{i-d(z_1,z_{-1})} \cdot (\frac{1-\rho'}{4})^{d(z_1,z_{-1})}$. 
 Note that we will consider the case when $i=0$ and hence $z_1$ and $z_{-1}$ is the empty string. In this scenario, $d(z_1 , z_{-1})$ is defined to be zero. We now consider all the above inequalities generated by choosing $(z_1, z_{-1}) \in \{-1,1\}^i \times \{-1,1\}^i$ for $0 \le i <n$.  Next, we add all these inequalities (Fact~\ref{fac:deg-increase}) but the degree of the resulting SoS proof remains $d_\gamma$.  
 Now, it is easy to see that all terms of the form : $\widetilde{J}(\mathbf{E} [g_{z_1}] , \mathbf{E} [g_{z_{-1}}])$ cancel out except when $z_1, z_{-1} \in \{-1,1\}^n$ or $z_1 = z_{-1}= \phi$.   $$
\mathop{\mathbf{E}}_{\substack{x \in \{-1,1\}^n  \\ y \sim_{\rho'} x}} [\widetilde{J}(g(x), g(y))]  \ge \widetilde{J}(\mathbf{E}[g(x)] ,  \mathbf{E}[g(y)]) + \textrm{error terms}
 $$
 
 We now compute the error terms. First, we sum up the error coming from the term $\epsilon \left( \widehat{g}_{z_1} ^2(n-i)+\widehat{g}_{z_{-1}} ^2(n-i)\right)$.  For any given $z_1 \in \{-1,1\}^i$, consider the term $\beta  \widehat{g}_{z_1} ^2(n-i)$. For every $z_{-1} \in \{-1,1\}^i$, it occurs with the factor $\frac{1+\rho'}{4}^{i-d(z_1,z_{-1})} \cdot \frac{1-\rho'}{4}^{d(z_1,z_{-1})}$. Since there are exactly $\binom{i}{k}$  strings $z_{-1} \in \{-1,1\}^i$ such that $d(z_1, z_{-1}) =k$, hence we get that the total weight associated is 
 $$
 \sum_{k=0}^i \binom{i}{k}\left(\frac{1+\rho'}{4}\right)^{i-k}\left(\frac{1-\rho'}{4}\right)^{k} = 2^{-i}
 $$
 Thus, we get that the first kind of error terms contribute $\epsilon \left( \widehat{g}_{z_1} ^2(n-i)+\widehat{g}_{z_{-1}} ^2(n-i)\right)$. The calculation of the ``fourth degree" error terms is exactly identical resulting in the final theorem. 
\end{proofof}
\begin{proofof}{(of Claim~\ref{clm:final-1})}
We have
\begin{align*}
 \vdash_6 \sum_{i=0}^{n-1}  \mathbf{E}_{x \in \{-1,1\}^{i}} [\widehat{g}_x^3(n-i) \widehat{h}_x(n-i)]   \le \frac{\sqrt{\eta}}{2} \left( \sum_{i=0}^{n-1} \mathbf{E}_{x \in \{-1,1\}^i} \widehat{g}_x^6(n-i) \right)  + \frac{ \left( \sum_{i=0}^{n-1} \mathbf{E}_{x \in \{-1,1\}^i} \widehat{h}_x^2(n-i) \right)}{2\sqrt{\eta}}  \end{align*}
Next, recall that using Fact~\ref{fac:temp}, we have $A_p \vdash_3 \sum_{i=0}^{n-1} \mathbf{E}_{x \in \{-1,1\}^i} \widehat{h}_x^2(n-i) \le \eta$. Similarly, from (\ref{eq:prelim-obs2}) and Claim~\ref{clm:subsec}, we have that $A_p \vdash -1\le \hat{g}_x(n-i) \le 1$. 
This in turn implies  $A_p \vdash_7 \widehat{g}_x^6(n-i) \le  \widehat{g}_x^2(n-i)$ (combining Fact~\ref{fac:SoS-10} and Fact~\ref{fac:SoS-2})

Combining all the above, we get 
$$
A_p \vdash_7 \sum_{i=0}^{n-1} \mathbf{E}_{x \in \{-1,1\}^i} [\widehat{g}_x^3(n-i) \widehat{h}_x(n-i)] \le \frac{\sqrt{\eta}}{2} \cdot \left(\sum_{i=0}^{n-1} \mathbf{E}_{x \in \{-1,1\}^i}\widehat{g}_x^2(n-i)\right) + \frac{\sqrt{\eta}}{2} 
$$
However, again we have that $A_p \vdash_3 \sum_{i=0}^{n-1} \mathbf{E}_{x \in \{-1,1\}^i}\hat{g}_x^2(n-i) = \sum_{|S|>0} \widehat{g}^2(S) \le \mathbf{E} [g^2(x)] \le 1$.  This gives us the claim. 
\end{proofof}
\begin{proofof}{(of Claim~\ref{clm:final-2})}
We have 
\begin{align*} 
 &\vdash_4 \sum_{i=0}^{n-1}  \mathbf{E}_{x \in \{-1,1\}^{i}} [\widehat{g}_x^2(n-i) \widehat{\ell}_x^2(n-i)] \le   \frac{\sqrt{\eta}}{2} \left( \sum_{i=0}^{n-1} \mathbf{E}_{x \in \{-1,1\}^i} \widehat{g}_x^4(n-i) \right)  + \frac{ \left( \sum_{i=0}^{n-1} \mathbf{E}_{x \in \{-1,1\}^i} \widehat{\ell}_x^4(n-i) \right)}{2\sqrt{\eta}}  \\
&\vdash_4    \sum_{i=0}^{n-1} \mathbf{E}_{x \in \{-1,1\}^i} \widehat{\ell}_x^4(n-i) \le 9^{d_\eta} \cdot \left(  \sum_{i=0}^{n-1} \left(\mathbf{E}_{x \in \{-1,1\}^i} \widehat{\ell}_x^2(n-i) \right)^2 \right) \  \textrm{ (using Fact~\ref{fac:hyper}) }
\end{align*}
  As in the proof of Claim~\ref{clm:final-1}, we can show $A_p \vdash_5 \sum_{i=0}^{n-1} \mathbf{E}_{x \in \{-1,1\}^i} \widehat{g}_x^4(n-i)  \le 1$.  
 Combining all the above, we get 
 $$
 A_p \vdash_5 \sum_{i=0}^{n-1} \mathbf{E}_{x \in \{-1,1\}^{i}} [\widehat{g}_x^2(n-i) \widehat{\ell}_x^2(n-i)] \le \sqrt{\eta}  + \frac{9^{d_\eta}}{\sqrt{\eta}} \left(  \sum_{i=0}^{n-1} \left(\mathbf{E}_{x \in \{-1,1\}^i} \widehat{\ell}_x^2(n-i) \right)^2 \right)
 $$
  Again, observe that
 \begin{eqnarray*}\vdash_2\mathbf{E}_{x \in \{-1,1\}^i} \widehat{\ell}_x^2(n-i) = \sum_{S \subseteq \{n-i , \ldots, n \} : n-i \in S} \widehat{\ell}^2(S) &\le&   \sum_{S \subseteq[n] : n-i \in S} \widehat{\ell}^2(S) =\sum_{S \subseteq[n] : n-i \in S : |S| \le d_\eta} \widehat{g}^2(S) \\&\le&  \sum_{S \subseteq[n] : n-i \in S : |S| \le d_\eta} \widehat{f}^2(S)  \le  \Inf_{n-i}^{\le d_{\eta}}(f)\end{eqnarray*}
 By using the second bullet of Fact~\ref{fac:SoS-basic0}, we can also get
 $$
 \vdash_4 \sum_{i=0}^{n-1} \left(\mathbf{E}_{x \in \{-1,1\}^i} \widehat{\ell}_x^2(n-i) \right)^2 \le \sum_{i=0}^{n-1}  \left( \Inf_{n-i}^{\le d_{\eta}}(f) \right)^2
 $$
 Combining these, we get the final result. 
\end{proofof}
\begin{proofof}{(of Claim~\ref{clm:final-3})}
 \begin{eqnarray*}
 \vdash_6 \sum_{i=0}^{n-1}  \mathop{\mathbf{E}}_{x \in \{-1,1\}^{i}} [\widehat{g}_x^2(n-i) \widehat{h}_x(n-i) \widehat{\ell}_x(n-i)] &\le&   \frac{1}{2\sqrt{\eta}} \left( \sum_{i=0}^{n-1} \mathop{\mathbf{E}}_{x \in \{-1,1\}^i} \widehat{h}_x^2(n-i) \right)  \\ &+& \frac{ \sqrt{\eta}\left( \sum_{i=0}^{n-1} \mathop{\mathbf{E}}_{x \in \{-1,1\}^i}\widehat{g}_x^4(n-i)  \widehat{\ell}_x^2(n-i) \right)}{2} \end{eqnarray*}
From the proof of Claim~\ref{clm:final-1}, we know that $A_p \vdash_3 \sum_{i=0}^{n-1} \mathop{\mathbf{E}}_{x \in \{-1,1\}^i} \widehat{h}_x^2(n-i)  \le \eta$. Thus, we get
$$
 \vdash_6 \sum_{i=0}^{n-1}  \mathop{\mathbf{E}}_{x \in \{-1,1\}^{i}} [\widehat{g}_x^2(n-i) \widehat{h}_x(n-i) \widehat{\ell}_x(n-i)] \le   \frac{\sqrt{\eta}}{2}   + \frac{ \sqrt{\eta}\left( \sum_{i=0}^{n-1} \mathop{\mathbf{E}}_{x \in \{-1,1\}^i}\widehat{g}_x^4(n-i)  \widehat{\ell}_x^2(n-i) \right)}{2}
$$
However, 
$$
\vdash_8 \sum_{i=0}^{n-1} \mathop{\mathbf{E}}_{x \in \{-1,1\}^i}\widehat{g}_x^4(n-i)  \widehat{\ell}_x^2(n-i)  \le\frac{ \sum_{i=0}^{n-1} \mathop{\mathbf{E}}_{x \in \{-1,1\}^i}\widehat{g}_x^8(n-i)  +  \sum_{i=0}^{n-1} \mathop{\mathbf{E}}_{x \in \{-1,1\}^i}\widehat{\ell}_x^4(n-i) }{2}
$$
Following the same proof as in the proof of Claim~\ref{clm:final-1}, we can show that 
$$
A_p \vdash_9 \sum_{i=0}^{n-1} \mathop{\mathbf{E}}_{x \in \{-1,1\}^i}\widehat{g}_x^8(n-i) \le 1
$$
Similarly, from the argument in the proof of Claim~\ref{clm:final-2}, we can show that 
$$
A_p \vdash_4  \sum_{i=0}^{n-1} \mathop{\mathbf{E}}_{x \in \{-1,1\}^i}\widehat{\ell}_x^4(n-i) \le 9^{d_{\eta}} \cdot \left(\sum_{i=0}^{n-1}  \left( \Inf_{n-i}^{\le d_{\eta}}(f) \right)^2 \right)
$$
Combining these, we have the proof. 
\end{proofof}

\end{document}